\documentclass[a4paper,11pt]{article}

\pdfoutput=1 
\usepackage{amsthm,amsmath,amssymb,amsfonts,mathtools}   
\usepackage{slashed,latexsym,multirow} 
\usepackage{graphicx,fullpage,xcolor}
\usepackage[makeroom]{cancel}
\usepackage[caption = false]{subfig}
\usepackage{enumitem,float}
\usepackage{physics, verbatim}

\usepackage{hyperref}
\hypersetup{colorlinks=true, linkcolor=blue, citecolor=red, urlcolor=blue}

\usepackage[nottoc,notlot,notlof]{tocbibind}

\usepackage[utf8]{inputenc}
\usepackage[T1]{fontenc}

\newtheorem{theorem}{Theorem}[section] 
\newtheorem{corollary}{Corollary}[theorem] 
\newtheorem{lemma}[theorem]{Lemma} 
\newtheorem{proposition}[theorem]{Proposition} 
\newtheorem{remark}[theorem]{Remark}

\newcommand{\ds}{{\slashed\partial}}
\newcommand{\opr}{{\rho^\circ}}
\newcommand{\crho}{{\rho^\circ}}
\newcommand{\omg}{{A_{(2)}}}
\newcommand{\A}{{\cal A}}
\newcommand{\D}{{\cal D}}
\newcommand{\HH}{{\cal H}}
\newcommand{\M}{{\cal M}}
\newcommand{\J}{\mathcal{J}}

\newcommand{\C}{{\mathbb C}}
\newcommand{\HHH}{{\mathbb H }}
\newcommand{\cinf}{{C^\infty(\M)}}
\newcommand{\I}{\mathbb I}

\DeclareMathOperator{\Ad}{Ad}

\begin{document} 

\title{\vspace{-.5truecm}{\bf A minimal twist for the Standard Model
    in noncommutative geometry I:  the field content}}
\author{Manuele Filaci\textsuperscript{*\ddag},  Pierre Martinetti\textsuperscript{\dag\ddag}, Simone Pesco\textsuperscript{*}
\\[5pt] 
\textsuperscript{*}\emph{Universit\`a di Genova -- Dipartimento di Fisica}, 
\textsuperscript{\dag}\emph{di Matematica}, \\ 
\textsuperscript{\ddag}\emph{INFN sezione di Genova,} \\[0pt]  
\emph{via Dodecaneso, 16146 Genova GE Italy.} \\[6pt] 
\emph{E-mail:} manuele.filaci@ge.infn.it, martinetti@dima.unige.it, simone.pesco.cpt@gmail.com}
\maketitle
\begin{abstract}
Noncommutative geometry provides both a unified description of the
Standard Model of particle physics together with Einstein-Hilbert
action (in euclidean signature) and some tools to go beyond the
Standard Model. In this paper, we extend to the full noncommutative
geometry of the Standard Model the twist (in the sense of
Connes-Moscovici) initially worked out for the electroweak sector and
the free Dirac operator only. Namely, we apply the twist also to the strong interaction
sector and the finite part of the Dirac operator. To do so, we
  are forced to take into
  account a violation of the twisted first-order condition.
As a result, we
still obtain the
 extra scalar field required to stabilise
the electroweak vacuum and fit  the Higgs mass, but it now has two chiral components. We also get the additive field
of $1$-forms already pointed out in the electroweak model, but with a
richer structure. Finally, we obtain a pair of Higgs
doublets, which are expected to combine into  a single Higgs doublet in
the action formula, as will be investigated in the second part of this work.
\end{abstract}
\tableofcontents

\section{Introduction}

Noncommutative geometry \cite{Connes:1994kx} (see \cite{Connes:2019ac}
for a recent review of the various aspects of the field) provides a mathematical framework in which a
single action formula yields both the lagrangian of the Standard
Model of fundamental interactions and the Einstein-Hilbert action 
(in
euclidean signature). As an added value, the Higgs field is
obtained on the same footing as the other gauge bosons -- as a
connection $1$-form -- but a
connection that lives on a slightly generalised notion  of
space, where points come equipped with an internal
structure. Such ``spaces'' are described by
\emph{spectral triples}
\begin{equation}
\label{eq:spectriple}
  \A,\; \HH,\;  D
\end{equation}
consisting in an algebra $\A$, acting on an Hilbert space $\HH$
together with an operator $D$ on $\HH$ which satisfies a set of axioms \cite{Connes:1996fu}
 guaranteeing that -- in case $\A$ is commutative and unital --  then there exists
 a (closed) riemannian spin manifold $\M$ such that $\A$ coincides with the
 algebra $\cinf$ of smooth functions on $\M$. In other terms, a
 spectral triple with $\A$ commutative does encode all the geometrical
information of a (closed) riemannian spin manifold \cite{connesreconstruct}. These axioms still
make sense when $\A$ is noncommutative, and provide then a definition
of a \emph{noncommutative geometry} as a spectral triple in which the
algebra is non necessarily commutative. 

The spectral triple of the Standard Model \cite{Chamseddine:2007oz} is built upon an
``almost-commutative algebra'',
\begin{equation}
\label{eq:almostcom}
\cinf\otimes \A_{\text{SM}}  
\end{equation}
where $\M$ is an even dimensional closed riemannian spin manifold and $\A_{\text{SM}}$ a noncommutative matrix algebra that encodes the gauge degrees of freedom
of the Standard Model. As explained in 
\cite{Connes:1996fu}, this non-commutative algebra provides the points
of $\M$ with an internal structure, in such a way that the Standard Model
is actually nothing but a pure theory of gravity, on a space that is made
slightly noncommutative by multiplying the (infinite dimensional)
commutative algebra $\cinf$ with the finite dimensional noncommutative $\A_{\text{SM}}$.

After the discovery of the Higgs boson in 2012, it has been noticed in
\cite{Chamseddine:2012fk} that an extra scalar field -- usually denoted
$\sigma$ --
proposed by particle physicist to cure the instability of the
electroweak vacuum due to the ``low mass of the Higgs'', also makes
the computation of the Higgs mass (which is not a free parameter in
the noncommutative description of the Standard Model) compatible with
its experimental value. Various scenarios have been proposed to make this extra scalar field emerge from the
mathematical framework of noncommutative geometry,  all of them
consisting in some modification of one of the axioms, the \emph{first
  order condition}
(e.g. \cite{Chamseddine:2013fk,Chamseddine:2013uq,T.-Brzezinski:2016aa,Brzezinski:2018aa,Boyle:2019ab,Besnard:2019aa,Besnard:2020ac},
see \cite{Chamseddine:2019aa} for a recent review).

In this paper, we push forward one of these scenarios, consisting in
twisting the spectral triple of the Standard Model. Twists have been introduced by Connes and Moscovici in
\cite{Connes:1938fk} with purely mathematical motivations. Later, it
has been discovered in \cite{buckley} that a very simple twist of
the Standard Model not
only produces the extra scalar field $\sigma$, but also  an additive
field of $1$-form $X_\mu$ which turns out to be related with Wick
rotation and the transition from the euclidean to the lorentzian
signature 
\cite{Devastato:2018aa,Martinetti:2019aa}. However, in \cite{buckley} the
twist was applied only to the part of the spectral triple that
yields the field $\sigma$, namely the subalgebra of $\A_\text{SM}$
describing the electroweak interaction and the part
of the operator $D$ that contains the Majorana mass of the
neutrinos. For simplicity, the subalgebra  of $\A_\text{SM}$
describing the strong interaction was left untouched, and  the part of
$D$ containing the Yukawa coupling of fermions was not taken into account.
In this paper,  we extend the
twisting procedure  to the whole spectral triple of the Standard
Model, according to the following lines.

The twist of gauge theories have been investigated in a systematic way
in \cite{Lett.,Landi:2017aa}, where the
twisted version of the first-order condition -- introduced by
imitation of the non-twisted case in \cite{buckley} -- has been put onto
solid mathematical bases. A notion
of \emph{minimal twist} of a spectral triple has also been
defined, which consists in making several copies of $\A$ act on $\HH$,
letting $D$ untouched. By doing so, one produces models
with new bosonic fields, keeping the fermionic content untouched, in
agreement with the state of the art of the
Standard Model (indeed the metastability of the electroweak vacuum points
towards new scalar fields, but there are no indications of new
fermions). A procedure for minimally twisting any real
spectral triple is to make two copies of the algebra act independently on the eigenspaces of the
grading operator. However,  applied to the Standard
Model, this does not produce any
extra-scalar field, as explained in  \cite{Manuel-Filaci:2020aa}.  

That is why in this paper we investigate another minimal twist of the Standard
Model, that does produces an extra scalar $\sigma$. The
price to pay is a violation of the twisted first order condition,
which is taken into account following the way pioneered in
\cite{Chamseddine:2013fk} and adapted to the twisted case in
\cite{Martinetti:2021aa}.

Besides the field content  of the Standard Model, we find that the extra
scalar $\sigma$ actually decomposes into two chiral
components $\sigma_r, \sigma_l$ (proposition \ref{proposi:offdiagfluc}) which are
invariant under a gauge transformation (proposition \ref{prop:gaugesigma}).
 We
also work out the structure of the $1$-form field $X_\mu$ (proposition
\ref{prop:freeunimod}), and study
how it behaves under a gauge transformation (proposition
\ref{prop:gaugetransform1}). In brief, imposing the same
unimodular condition as in the non-twisted case, we find that the
antiselfadjoint part of the (generalised)  $1$-form generated by the free
Dirac operator $\ds$ yields exactly the bosonic
content of the Standard Model, as in the non-twisted case. But there is
also a selfadjoint part made of two real $1$-form fields and one
selfadjoint $M_3(\C)$-value $1$-form field. Altogether these three
fields compose the $1$-form field $X_\mu$.

 The
complete understanding of the physical
meaning of these fields passes through the computation of the fermionic and
spectral actions, and will be the object of a second paper
\cite{Manuel-Filaci:2020aa}.
\smallskip

The paper is organised as follows. In section \ref{sec:nontwisted} we
recall the basics of the spectral triple of the Standard Model
(\S\ref{subsec:sm}), make explicit the tensorial notations employed all
along the paper (\S\ref{subsec:representation}) and use them to write
explicitly the Dirac operator, the grading and the real structure (\S\ref{subsec:Dirac}).
Section \ref{sec:minimaltwist} deals with the twist. After recalling
the procedure of minimal twisting defined in \cite{Lett.}, we apply it
to the spectral triple of the Standard Model: the algebra is doubled
so as to act independently on the left and right components of Dirac
spinors (\S \ref{subse:twistalgH}). The grading and the real structure
are the same as in the non-twisted case, and we check explicitly that
one of the axioms (the order zero condition) still holds in the
twisted case (\S\ref{subsec:algopp}), as expected from the general
result of \cite{Lett.}. Paragraph \ref{subsec:twistfluc}
is a brief recalling about twisting fluctuations, that is the way to
generate the bosonic fields. The detail computation of this fluctuations is the
object of section \ref{sec:scalarsector} and \ref{sec:gauge}, which contain the main
results of this paper. We first work out the Higgs
sector in \S \ref{lem:diag1form}. The main result is proposition
\ref{prop:diagfluc} in which we find two Higgs doublets.  The extra
scalar field $\sigma$ is generated in \S\ref{subsec:sigma}. Its
structure as a doublet of real scalar fields $\sigma_r, \sigma_l$ is
established in proposition
\ref{proposi:offdiagfluc}. In section \ref{sec:gauge} we compute the
twisted fluctuation of the free part $\ds$  of the Dirac
operator. Useful properties of the Dirac matrices with respect to the
twist are worked out in \S\ref{subsec:gamma}.   The generalised
twisted $1$-forms generated by the free Dirac operator are computed in
\S\ref{sec:free1form}, and the physical degrees of freedom are
identified in \S\ref{subsec:physdegrees}. The structure of the
$1$-form field $X_\mu$ is
summarised in proposition \ref{prop:freeunimod} and yields, in
\S\ref{subsec:twistfreefluc}, the explicit form of the twisted fluctuation
of the free Dirac operator. In section \ref{sec:gaugetransform} we
study how all these fields behave under a gauge transformation. After
recalling the basics of gauge transformation for twisted spectral
triple (as stabilised in \cite{Landi:2017aa}), we apply these techniques
to the gauge and the $1$-form fields in
\S\ref{subsec:gaugegauge}  and to the scalar fields in
\S\ref{gauge:scalar}. We show in proposition \ref{prop:gaugetransform1} that the bosonic fields
transform in the correct way, while the $1$-form field is invariant,
up to a unitary transformation on the $M_3(\C)$-value part. The Higgs
doublets as well transform as expected (proposition \ref{prop:gaugehiggs}), while the extra
scalar field $\sigma$ is gauge invariant, as shown in proposition \ref{prop:gaugesigma}.

The first appendix contains notations and generalities on Dirac
matrices. In the second one, we write explicitly the components of the
twisted fluctuation in terms of the gauge fields (this will be useful
in the second part of the paper, to compute the action). In the last
appendix, we check that the twisted first-order condition is only
partially verified.
 \bigskip

{\bf Notations and important comments regarding the literature:}
\begin{itemize}
\item [-]

In the first version of this paper,
we erroneously thought the twist we were using was ``by grading'', and assumed
 the twisted first-order condition. Actually the
latter is violated only by the off-diagonal part of the internal Dirac
operator, and this does not modify the extra-scalar
field, as explained before remark \ref{rem:extra-scalar-field} neither
the gauge invariance of
the fermionic action, as explained before proposition \ref{prop:gaugesigma}.
 
\item[-] We work with one generation of fermions (electron $e$,
  neutrino $\nu_e$, quarks up $u$ and down~$d$). The extension 
to three generations will be
  discussed in the second part of the work \cite{Filaci:2021aa}.

\item[-] All along the paper, we apply the usual rule of contractions of
indices in alternate up/down positions. Typically the greek indices
label the coordinates of the manifold. 
\end{itemize}


\section{The non-twisted case}
\label{sec:nontwisted}
\setcounter{equation}{0} 

As a preparation to the twisting, we recall in this section the main
features of the spectral description of the Standard Model. Besides
the original papers (recalled in the text),  the details are extensively discussed in the books
\cite{Connes:2008kx} and \cite{Walterlivre} (for a more
physics-oriented presentation).

\subsection{The spectral triple of the Standard Model}
\label{subsec:sm}

The usual spectral triple of the Standard Model
\cite{Chamseddine:2007oz} is
the product
 of the canonical triple
of a (closed) riemannian spin manifold $\M$ of even dimension $m$,
\begin{equation}
  \label{eq:33}
  \cinf, \quad L^2(\M, S),\quad \ds
\end{equation}
with the finite dimensional spectral triple (called \emph{internal})
\begin{equation}
  \label{eq:34}
  \A_{\text SM}= \C \oplus \HHH \oplus M_3(\C),\quad \HH_F= \C^{32n},\quad D_F
\end{equation}
that describes the gauge degrees of
freedom of the Standard Model. 
In \eqref{eq:33}, $\cinf$ denotes the algebra of
smooth functions on $\M$, that acts by multiplication on the Hilbert space  $L^2(\M, S)$  of square
integrable spinors as
\begin{equation}
  \label{eq:106}
  (f\psi)(x) = f(x)\psi(x)\quad\forall f\in\cinf, \, \psi\in L^2(\M, S),\, x\in \M,
\end{equation}
while
\begin{equation}
\slashed{\partial}=-i\gamma^\mu\nabla_\mu\quad \text{ with }\quad \nabla_\mu = \partial_\mu+ \omega_\mu
\label{eq:101}
\end{equation}
is the Dirac operator on $L^2(\M, S)$ 
associated with the spin connection $\omega_\mu$ and the $\gamma^\mu$s
are the Dirac matrices associated with the riemannian metric $g$ on
$\M$:
\begin{equation}
  \label{eq:110}
  \gamma^\mu\gamma^\nu + \gamma^\nu\gamma^\mu = 2g^{\mu\nu}\mathbb I \quad
  \forall \mu, \nu = 0, m-1
\end{equation}
($\I$ is the identity operator on $L^2(\M, S)$ and we label the
coordinates of $\M$ from $0$ to $m-1$).

In \eqref{eq:34}, $n$ is
the number of generations of fermions, and $D_F$ is a $32n$ square
complex matrix whose entries are the Yukawa couplings of fermions and
the coefficients of the Cabibbo-Kobayashi-Maskawa (CKM) mixing matrix of quarks and of the Pontecorvo-Maki-Nakagawa-Sakata (PMNS) mixing matrix of neutrinos.
Details are given in \S \ref{subsec:Dirac}, and the representation of
$\A_{\text{SM}}$ on $\HH_F$ is~in~\S\ref{subsec:representation}.

The product spectral triple is 
\begin{equation}
\cinf\otimes \A_{\text{SM}}, \quad \HH= L^2(\M, S)\otimes \HH_F,\quad
D=\ds\otimes \I_F + \gamma_\M\otimes D_F
\label{eq:04}
\end{equation}
with $\gamma_\M$  the product of the euclidean Dirac matrices (appendix \ref{app:Dirac}) and
$\I_F$ the identity~on~$\HH_F$.

A spectral triple $(\A, \HH, D)$ is \emph{graded} when the Hilbert space comes equipped
with a grading (that is a selfadjoint operator that squares to $\I$)
which anticommutes with $D$. The spectral triple
\eqref{eq:33} is graded with grading $\gamma_\M$. The internal spectral triple \eqref{eq:34} is graded, with grading
the operator $\gamma_F$ on $\HH_F$ that  takes value $+1$ on right particles \&  left
antiparticles, $-1$ on left particles \&  right
antiparticles. 
The product spectral triple \eqref{eq:04} is graded, with grading
\begin{equation}
  \label{eq:14}
  \Gamma=\gamma_\M\otimes\gamma_F.
\end{equation}

Another important ingredient is the \emph{real
  structure}, that is an anti-linear operator that squares to $\pm\I$ and
commutes or anticommutes with the grading and the operator $D$ (the
possible choices define the so-called $KO$-dimension of the spectral
triple).  For a manifold, the real structure $\J$ is given by the
charge conjugation operator. In dimension $m=4$, it satisfies
\begin{equation}
  \label{eq:103}
  {\cal J}^2 = -\I,\quad {\cal J} \ds = \ds{\cal J}, \quad  {\cal J} \gamma_\M = \gamma_\M{\cal J}.
\end{equation}
The real structure of the internal spectral
triple \eqref{eq:34} is the anti-linear operator $J_F$ that exchanges particles with antiparticles
on $\HH_F$. It satisfies
\begin{equation}
  \label{eq:104}
  J_F^2= \I, \quad J_F D_F = D_F J_F, \quad J \gamma_F =- \gamma_F J_F.
\end{equation}
The real structure for the product spectral triple \eqref{eq:04} is
\begin{equation}
  \label{eq:20}
  J=\J \otimes J_F.
\end{equation}
For a manifold of dimension $m=4$, it is such that
\begin{equation}
  \label{eq:105}
  J^2= -\I, \quad JD = DJ,\quad J\Gamma = -\Gamma J.
\end{equation}

The real structure implements an action of the opposite algebra
$\A^\circ$ on $\HH$, identifying $a^\circ\in\A^\circ$ with
$Ja^*J^{-1}$. This action is asked to commute with the one of $\A$,
yielding the \emph{order~zero~condition}
\begin{equation}
  \label{eq:57}
  [a, b^\circ]=0 \quad \forall a\in \A, b\in \A^\circ.
\end{equation}
Among the properties of a spectral triple, one particularly relevant
for physical models is the
\emph{first order condition}
\begin{equation}
  \label{eq:13}
 [ [D, b], a^\circ]=0 \quad \forall a, b\in \A.
\end{equation}
\newpage

\subsection{Representation of the algebra}
\label{subsec:representation}

 To describe
the action of $\A_\text{SM}\otimes \cinf$ on $\HH$ in \eqref{eq:04}, it is convenient to
label the $32n$ degrees of freedom of the finite dimensional Hilbert
space $\HH_F$ by a multi-index $\,C\,I\,\alpha$
defined as follows. 

  \begin{itemize}
  \item[$\bullet$] $C=0,1$ is for particle $(C=0)$ or anti-particle
    $(C=1)$;
  \item[$\bullet$] $I=0;\, i$ with $i=1,2,3$ is the lepto-colour index: $I=0$
    means lepton, while $I=1, 2, 3$ are for the quark, which exists in
    three colours;
  \item[$\bullet$] $\alpha=\dot{1},\dot{2};a$ with $a=1,2$ is the flavour index:
    \begin{align}
      &\dot{1}=\nu_R, \; \dot{2}=e_R,\; 1=\nu_L,\; 2=e_L&&\text{for leptons ($I=0$)},\\
      &\dot{1}=u_R, \; \dot{2}=d_R,\; 1=q_L,\; 2=d_L&&\text{for quarks ($I=i)$}.
    \end{align}
\noindent We sometimes use the shorthand notation
$\ell_L^a=\left(\nu_L,e_L\right)$ for the left handed neutrino and the
associated lepton, and $q_L^a=\left(u_L,d_L\right)$ for the pair of
left-handed quarks.
  \end{itemize}
There are $2\times 4 \times 4=32$ choices of triplet
  of indices $(C, I, \alpha)$, which is the number of fermions per
  generation. One should also take into account an extra index $n=1,2,3$ for the generations,
  but in this paper we work with one generation only and we
  omit it (we will discuss the
  number of generations in the computation of the action \cite{Manuel-Filaci:2020aa}). So from now on 
  \begin{equation}
    \label{eq:51}
    \HH_F=\C^{32}.
  \end{equation}
An element $\psi\in\HH=\cinf\otimes \HH_F$
is thus a $32$ dimensional column-vector, in which each
component $\psi_{CI\alpha}$ is a Dirac spinor in $L^2(\M, S)$. 

Regarding the algebra, unless necessary we omit the symbol of the representation and
identify an element $a~=~(c,q,m)$ in $C^\infty(\M)\otimes\A_\text{SM}$,
where
\begin{equation}
  \label{eq:61}
c\in C^\infty(\M, \C),\quad q\in C^\infty(\M, \HHH), \quad
  m\in C^\infty(\M, M_3(\C)),
\end{equation}
with its representation as bounded
operator on $\HH$, that is a $32$ square matrix whose components{\footnote{$D, J, \beta$ are column indices with the same range as
the line indices $C, I, \alpha$ (the position of the indices was slightly different
    in \cite{buckley}, the one adopted here makes the
    tensorial computation more tractable).}} 
\begin{equation}
a_{CI\alpha}^{DJ\beta}
\label{eq:107}
\end{equation}
are smooth functions acting by multiplication
on $L^2(\M, S)$ as in \eqref{eq:106}.
 Explicitly{\footnote{The indices after the closing parenthesis are
    here to recall that the block-entries of $\A$ are labelled by the
    $C, D$ indices, that is $a^1_1= Q, a_2^2=M, a_1^2=a_2^1 = 0$.}} 
\begin{equation}
  \label{eq:27}
  a=\left(\begin{matrix}
Q&\\ & M 
\end{matrix}\right)_C^D
\end{equation}
where the $16\times 16$ square matrices $Q$, $M$ have components
 \begin{equation}
   \label{eq:59}
   Q_{I\alpha}^{J\beta}=\delta_I^JQ_\alpha^\beta ,\quad
   M_{I\alpha}^{J\beta}= \delta_{\alpha}^{\beta}M_I^J,
 \end{equation}
where 
\begin{equation}
  \label{eq:60}
  Q_\alpha^\beta=
  \begin{pmatrix}
    c& & \\ & \bar c&\\ &&q
  \end{pmatrix}_\alpha^\beta,\quad M_{I}^{J}=
  \begin{pmatrix}
    c& \\ &m
  \end{pmatrix}_I^J.
\end{equation}
Here, the over-bar $\bar{\cdot }$, denotes the complex conjugate, $m$
(evaluated at the point $x$) identifies with its usual representation as
$3\times 3$ complex matrices and the quaternion
$q$ (evaluated at $x$)
acts through its representation as $2\times 2$ matrices:
\begin{equation}
  \HHH\ni q(x)=
  \begin{pmatrix}
    \alpha& \beta \\ -\bar\beta & \bar\alpha
\end{pmatrix}, \qquad \alpha, \beta\in\C.
\end{equation}

\subsection{Finite dimensional Dirac operator, grading and real structure}
\label{subsec:Dirac}

With respect
to the particle/antiparticle index $C$, the internal Dirac operator
 \begin{equation}
D_F=D_Y+D_M
\label{eq:8}
\end{equation}
decomposes into a diagonal and an off-diagonal part
\begin{equation}
D_Y=\left(\begin{matrix}
D_0&\\&D_0^\dagger
\end{matrix}\right)_C^D, \quad D_M= \left(\begin{matrix}
0&D_R\\D_R^\dagger&0
\end{matrix}\right)_C^D
\label{eq:9}
\end{equation}
containing respectively the Yukawa couplings of fermions and the
Majorana mass of the neutrino.

The $16\times 16$ matrices $D_0$ and $D_R$ are block-diagonal with
respect to the lepto-colour index~$I$
\begin{equation}
\label{eq:D0}
D_0=\left(\begin{matrix}
D_0^\ell&&&\\&D_0^q&&\\&&D_0^q&\\&&&D_0^q
\end{matrix}\right)_I^J,\quad D_R=
  \begin{pmatrix}
    D_R^\ell & & &\\ &0_4& & \\ &&0_4&\\&&&0_4
  \end{pmatrix}_I^J,
\end{equation}
where we write $\ell$ for $I=0$ and $q$ for $I=1,2,3$. Each $D_0^I$ is a $4\times 4$ matrix (in the flavour index~$\alpha$),
\begin{equation}
\label{eq:defD0I}
D_0^I=\left(\begin{matrix}
0&\overline{\sf k^I}\\
\sf k^I&0
\end{matrix}\right)_\alpha^\beta \quad \text{ where }\quad  {\sf k^I} \coloneqq \left(\begin{matrix}
k_u^I&0\\
0&k_d^I 
\end{matrix}\right)_\alpha^\beta,
\end{equation} whose entries are the Yukawa couplings of elementary fermions
\begin{equation}
k_u^I=\left(\begin{matrix}
k_\nu,&k_u,&k_u,&k_u
\end{matrix}\right)\\
\qquad k_d^I=\left(\begin{matrix}
k_e,&k_d,&k_d,&k_d
\end{matrix}\right)\\
\label{eq:10}
\end{equation}
(three of them are equal because the Yukawa coupling of quarks does
not depend on the colour). Similarly,  $D_R^\ell$ is a $4\times 4$
matrix (in the flavour index),  
\begin{equation}
  \label{eq:11}
  D_R^\ell =
  \begin{pmatrix}
    k_R& \\ &0_3
  \end{pmatrix}_\alpha^\beta
\end{equation}
whose only non-zero entry is the Majorana mass of the neutrino.

In tensorial notations, one has
\begin{equation}
\label{eq:Drtensor}
D_R=k_R\, \Xi_{I\alpha}^{J\beta}
\end{equation} 
where
\begin{equation}
  \label{eq:76}
  \Xi_\alpha^\beta\coloneqq 
  \begin{pmatrix}
    1 & \\ &0_3
  \end{pmatrix}_\alpha^\beta,\quad  \Xi_I^J\coloneqq 
  \begin{pmatrix}
    1 & \\ & 0_3
  \end{pmatrix}_I^J
\end{equation}
and $\Xi_{I\alpha}^{J\beta}$ is a shorthand notation for the tensor $\Xi_{I}^{J}\Xi_{\alpha}^{\beta}$.
Similarly, the internal grading is 
\begin{equation}
\gamma_F=
\begin{pmatrix}
   \I_8 & & &\\  
& -\I_8 &  &\\
 &&  -\I_8 &\\  
&&& \I_8  \end{pmatrix}=\eta_{C\alpha}^{D\beta}\delta_I^J
\end{equation}
where the blocks in the matrix act respectively on right/left
particles, then right/left antiparticles, and we define
\begin{equation}
  \label{eq:17bbis}
\eta_\alpha^\beta \coloneqq   \begin{pmatrix}
    \I_2 & \\ &-\I_2
  \end{pmatrix}_\alpha^\beta,\quad \eta_C^D \coloneqq   \begin{pmatrix}
    1 & \\ &-1
  \end{pmatrix}_C^D
\end{equation}
and $\eta_{C\alpha}^{D\beta}$ holds for $\eta_C^D\eta_\alpha^\beta$.  
The internal real structure is
\begin{equation}
J_{F}=\left(\begin{array}{cc}
0 & \mathbb{I}_{16}\\
\mathbb{I}_{16} & 0
\end{array}\right)_C^D cc= \xi_C^D\,\delta_{I\alpha}^{J\beta} cc
\label{eq:2bis}
\end{equation}
where $cc$ denotes the complex conjugation and we define
\begin{equation}
  \label{eq:37bis}
  \xi_C^D\coloneqq 
  \begin{pmatrix}
    0&1\\1&0
  \end{pmatrix}_C^D.
\end{equation}

\section{Minimal twist of the Standard Model}
\label{sec:minimaltwist}
\setcounter{equation}{0} 

In the noncommutative geometry description of the Standard Model, the
bosonic degrees of freedom are obtained by a so-called fluctuation of
the metric, that is the substitution of the operator $D$ with $D + A +
J AJ^{-1}$ where
\begin{equation}
\label{eq/fluct}
A=\sum_i a_i[D, b_i] \qquad  a_i, b_i \in \A
\end{equation}
is a generalised $1$-form (see
\cite{Connes:1996fu} for details and the justification of the
terminology).
 
As already noticed in \cite{Chamseddine:2007oz,Connes:2008kx}, the Majorana mass of the neutrino does not contribute to the bosonic
content of the model, for $D_M$ commute with algebra:
\begin{equation}
\label{eq:needfortwist}
 [\gamma^5\otimes D_M, a]=0 \quad \forall a\in \A.
\end{equation}
However, in order to generate the $\sigma$ field proposed in
\cite{Chamseddine:2012fk} to cure the electroweak vacuum instability and solve the problem of the computation of the Higgs mass,
one precisely needs to make $D_M$ contribute to the fluctuation.

To do this, a possibility consists in substituting the commutator
$[D, a]$ with a twisted commutator
\begin{equation}
\label{eq:needfortwist1}
   [D, a]_\rho:=D a- \rho(a)D
\end{equation}
where $\rho$ is a fixed automorphism of $\A$. 
This substitution is
the base of the definition of \emph{twisted spectral triple}
\cite{Connes:1938fk} where, instead of asking that $[D,a]$ be bounded
for any $a$ (which is one of the axioms of a spectral triple), one
requires that there exists an automorphism $\rho$ such that the
twisted-commutator $[D, a]_\rho$ is bounded
for any $a\in\A$. As shown in \cite{Lett.}, starting with a spectral triple $(\A, \HH, D)$ where
$\A$ is almost commutative as in \eqref{eq:almostcom}, then the only ways to build a twisted
spectral triple with the same Hilbert space and Dirac operator  (which, from a physics
point of view, means that one looks for models with the same fermionic
content as the Standard Model) are to double the algebra and make them
act independently on the left and right components of spinors
(following actually an idea of \cite{Devastato:2013fk}). All
this is detailed in the next section.

 \subsection{Algebra and Hilbert space}
\label{subse:twistalgH} 

The algebra $\A$ of the twisted spectral triple of the Standard Model is
twice the algebra \eqref{eq:04}, 
\begin{equation}
  \label{eq:3}
  \A=\left(\cinf\otimes \A_{\text{SM}}\right)\otimes \C^2,
\end{equation}
which is isomorphic to 
\begin{equation}
  \label{eq:4}
 \left(\cinf\otimes \A_{\text{SM}}\right)\oplus \left(\cinf\otimes \A_{\text{SM}}\right). 
\end{equation}

It acts on the same Hilbert space $\HH$ as in the non-twisted case, but now the two copies of $\cinf\otimes \A_{\text{SM}}$ act
independently on the right and left components of spinors.
To write this action, it is convenient to  view an element of $\HH$ as a column vector with $4\times
32=128$ components ($4$ being the number of components of a usual
spinor in $L^2(\M,S)$ for $m =4$). To this aim, one introduce two
extra-indices to label the degrees of freedom of 
$L^2(\M, S)$:
  \begin{itemize}
  \item $s=r,l$ is the chirality index;
  \item $\dot{s}=\dot{0},\dot{1}$ denotes particle ($\dot{0}$) or
    anti-particle part ($\dot{1}$).
  \end{itemize}

An element $a$ of \eqref{eq:4} is a pair of elements of
\eqref{eq:04}, namely
\begin{equation}
  \label{eq:52}
  a=(c, c', \, q, q',\, m, m') 
\end{equation}
with
\begin{equation}
  \label{eq:28}
 c, c'\in C^\infty(\M, \C)\quad q, q'\in
  C^\infty(\M, \HHH),\quad  m, m'\in C^\infty(\M, M_3(\C)).
\end{equation} 

We make $(c,q,m)$ act on the chiral subspace $\HH_c$ of $\HH$,
consisting in particles and antiparticles whose chirality as Dirac
spinors coincides with chirality in the internal space; whereas $(c',
q', m')$ acts on the anti-chiral subspace $\HH_a$ consisting in
particles and particles whose Dirac and internal chiralities do not
coincide. The chiral subspace $\HH_c$ is the subspace of $\HH$ spanned by $r, \alpha=\dot 1, \dot 2$ and $l, \alpha=1,2$, while $\HH_a$
is spanned by $l, \alpha  =\dot 1, \dot 2$ and $r, \alpha= 1, 
2$ (in both cases, $C$ takes both values $1,0$). In other terms, $a\in\A$ acts as in \eqref{eq:27}, but now
the two $64\times 64$  matrices $Q$, $M$ are tensor fields of components 
\begin{equation}
\label{eq:QMexplicit}
Q_{\dot s sI\alpha}^{\dot t t J\beta}
=\delta_{\dot{s}I}^{\dot{t}J}\, Q_{s\alpha}^{t\beta} ,\qquad       
M_{\dot s sI\alpha}^{\dot t t J\beta}=
\delta_{\dot{s}}^{\dot{t}} \, M_{s\alpha I}^{t\beta J} 
\end{equation}
where $\delta_{sI}^{tJ}$ denotes the product of the two
Kronecker symbols $\delta_s^t$, $\delta_I^J$.
Both $Q$ and $M$ stills 
act trivially (i.e. as the identity) on the indices $\dot s\dot t$,
but no longer on the chiral indices $st$. 
On the latter, the action is given by
\begin{equation}
\label{eq:defQM}
 Q_{s\alpha}^{t\beta}=\left(\begin{matrix}
(Q_r)_\alpha^\beta&\\ &(Q_l)_\alpha^\beta
\end{matrix}\right)_s^t ,
\qquad M_{s\alpha I}^{t\beta J}=\left(\begin{matrix}
(M_r)_{\alpha I}^{\beta J}&\\ & (M_l)_{\alpha I}^{\beta J}
\end{matrix}\right)_s^t, 
\end {equation}
with
\begin{align}
\label{eq:defQ}
Q_{r}=\left(\begin{matrix}
{\sf c}&\\ &q'
\end{matrix}\right)_\alpha^\beta ,\quad Q_l=\left(\begin{matrix}
{\sf c'}&\\  &q
\end{matrix}\right)_\alpha^\beta,
\end{align}
and 
\begin{align}
\label{eq:defM}
M_{r}=\begin{pmatrix}
    \sf m\otimes \I_2 & 0 \\
    0 &\sf m'\otimes \I_2
  \end{pmatrix}_\alpha^\beta,\; 
M_{l}=\begin{pmatrix}
    \sf m' \otimes \I_2 & 0 \\
    0 &\sf m\otimes \I_2
  \end{pmatrix}_\alpha^\beta,
\end{align}
where we denote
\begin{equation}
  \label{eq:35}
  {\sf c}\coloneqq 
  \begin{pmatrix}
    c & \\ & \bar c
  \end{pmatrix},\quad {\sf m}\coloneqq 
  \begin{pmatrix}
    c& \\ &m
  \end{pmatrix}_I^J,\; {\sf c'}\coloneqq 
  \begin{pmatrix}
    c' & \\ & \bar c'
  \end{pmatrix},\quad {\sf m'}\coloneqq 
  \begin{pmatrix}
    c'& \\ &m'
  \end{pmatrix}_I^J,\;
\end{equation}
Compared to the usual spectral triple of the Standard
Model, $M_{r/l}$ are no longer trivial in the flavour index~$\alpha$.
\begin{remark}
If we were using the twist-by grading, we should permute $\sf m$
  with  $\sf m'$ in \eqref{eq:defM}, for on the antiparticles subspace
 -- i.e.  $C=1$ -- then $\HH_c$ is a subspace of the $-1$-eigenspace of
 the grading (see also appendix \ref{sec:twistfirstorder} regarding the
  twist used in \cite{buckley}).
\end{remark}

The twist $\rho$ is the automorphism of $\A$ that exchanges the two
components of $\A_{\text{SM}}$, namely
\begin{equation}
  \label{eq:2}
  \rho(c, c', q, q', m, m')= (c', c, q', q, m', m).
\end{equation}
In terms of the representation, one has
\begin{equation}
  \label{eq:5}
  \rho(a) =
  \begin{pmatrix}
   \rho(Q) & \\ & \rho(M)
  \end{pmatrix}_C^D
\end{equation}
with
\begin{equation}
  \label{eq:22}
  \rho(Q)_{s\dot s I \alpha}^{t \dot t J \beta} = \delta_{\dot s
    I}^{\dot tJ} \, \rho(Q)_{s\alpha}^{t\beta}, 
\qquad \rho(M)_{s\dot s I \alpha}^{t\dot tJ \beta} =
\delta_{\dot s}^{\dot t} \, \rho(M)_{s\alpha I}^{t\beta J}
\end{equation}
where
\begin{equation}
  \label{eq:29}
  \rho(Q)_{s\alpha}^{t\beta}=
  \begin{pmatrix}
    (Q_l)_\alpha^\beta&\\ &  (Q_r)_\alpha^\beta
  \end{pmatrix}_s^t\; , \qquad 
 \rho(M)_{s\alpha I}^{t\beta J}=
  \begin{pmatrix}
    (M_l)_{\alpha I}^{\beta J} & \\ & (M_r)_{\alpha I}^{\beta J}
  \end{pmatrix}_s^t\;.
\end{equation}

\noindent In short, the twist amounts to flipping the left/right indices $l/r$.

\subsection{Grading and real structure}
\label{subsec:algopp}

The operators $\Gamma$ in \eqref{eq:14} and $J$
in \eqref{eq:20} are the grading and the real structure for the twisted spectral
triple, in the sense defined in \cite{buckley,Lett.} (the rule of
signs defining the $KO$-dimension is not affected by the twist; that
$\Gamma$ commutes with the representation \eqref{eq:defQM} follows
from the latter being diagonal but on the $\alpha$ and $I$ indices,
where $\Gamma$ is (block)-diagonal).
In particular, as in the non twisted case, the real structure implements an action of the opposite algebra
$\A^\circ$ on $\HH$, that commutes with the one of $\A$. To check this, let us first write
down the representation of the opposite algebra.
\begin{proposition}
\label{prop:algopp}
  For $a\in\A$ as in \eqref{eq:27}, one has (for $\M$ of dimension $4$)
  \begin{equation}
    \label{eq:55}
    JaJ^{-1}=
    -\begin{pmatrix}
      \bar M&0 \\ 0&\bar Q
    \end{pmatrix}_C^D.
  \end{equation}
\end{proposition}
\begin{proof}
From \eqref{eq:20} and \eqref{eq:2bis} one has
\begin{equation}
  \label{eq:108}
  J=\begin{pmatrix}
    0 &\J\otimes \I_{16}\\ \J\otimes \I_{16} & 0
  \end{pmatrix}_C^D.
\end{equation}
Since $J^{-1}= -J$ by \eqref{eq:105}, using the representation
\eqref{eq:27} of $a$ one obtains (omitting $\I_{16}$)
\begin{align}
  \label{eq:50}
  J aJ^{-1}= -JaJ &=-\begin{pmatrix} 0 & {\cal J}\\ {\cal J}& 0\end{pmatrix}_C^E
\begin{pmatrix}
Q& 0\\ 0& M
\end{pmatrix}_E^F
\begin{pmatrix}
0 & {\cal J}\\ {\cal J}& 0
\end{pmatrix}_F^D= -\begin{pmatrix} {\cal J} M {\cal J}&0 \\ 0&{\cal J}Q {\cal J}\end{pmatrix}_C^D.
\end{align}

In addition, $\J$ commutes with the grading $\gamma_\M$ (see
\eqref{eq:103}), so it is of the
form
\begin{equation}
  \label{eq:109}
  \J =
  \begin{pmatrix}
    \J_r &0 \\ 0&\J_l 
  \end{pmatrix}_s^t\, cc
\end{equation}
where $\J_{r\slash l}$ are $2\times 2$ matrices carrying the $\dot s, \dot
t$ indices, such that $\J_r \bar\J_r =\J_l \bar\J_l=-\mathbb I_2$ . From the explicit form \eqref{eq:QMexplicit} of $Q$
and $M$, one gets (still omitting the indices $\alpha, I$ in
which $J$ is trivial)
\begin{align}
\label{eq:JQJ}
{\cal J} Q {\cal J}= \begin{pmatrix}
    \J_r (\delta_{\dot s}^{\dot t}\bar Q_r) \bar\J_r&0 \\ 0&\J_l  (\delta_{\dot s}^{\dot t}\bar Q_l) \bar\J_l
  \end{pmatrix}_s^t=\begin{pmatrix}
   -\delta_{\dot s}^{\dot t}\bar Q_r&0 \\ 0&-\delta_{\dot s}^{\dot t}\bar Q_l
  \end{pmatrix}_s^t=-\bar Q,\\
\label{eq:JMJ}
{\cal J} M {\cal J}= \begin{pmatrix}
    \J_r (\delta_{\dot s}^{\dot t}\bar M_r) \bar\J_r&0 \\ 0&\J_l  (\delta_{\dot s}^{\dot t}\bar M_l) \bar\J_l
  \end{pmatrix}_s^t=\begin{pmatrix}
   -\delta_{\dot s}^{\dot t}\bar M_r&0 \\ 0&-\delta_{\dot s}^{\dot t}\bar M_l
  \end{pmatrix}_s^t=-\bar M,
\end{align}
hence the result.
\end{proof}

To check the order zero condition, we denote  
\begin{equation}
b=(d, d', p, p', n, n')\label{eq:15}
\end{equation}
another element of  $\A$
with $d, d'\in C^\infty(\M, \C)$, $p, p'\in C^\infty(\M, \HHH)$,
$n, n'\in C^\infty(\M, M_3(\C))$. It acts on $\HH$ by \eqref{eq:defb}
as
\begin{equation}
\label{eq:defb}
b=\left(\begin{matrix}
R&\\&N
\end{matrix}\right)_C^D
\end{equation}
where $R, N$ are defined as
 $Q, M$ in \eqref{eq:QMexplicit}, with 
\begin{equation}
R_r=\left(\begin{matrix} \sf d&\\ &p'
\end{matrix}\right)_\alpha^\beta\!\!,\; R_l=\left(\begin{matrix} \sf d'&\\ &p
\end{matrix}\right)_\alpha^\beta\!\!,\;  N_r=\left(\begin{matrix}
\sf n \otimes \I_2&\\ &\sf n'\otimes \I_2
\end{matrix}\right)_\alpha^\beta\!\!,\; N_l=\left(\begin{matrix}
\sf n' \otimes \I_2&\\ &\sf n\otimes \I_2
\end{matrix}\right)_\alpha^\beta\!\!.
\label{eq:31}
\end{equation}

\begin{corollary}
  The order-zero condition \eqref{eq:57} holds.
 \end{corollary}
 \begin{proof}
By Prop. \ref{prop:algopp},  the order zero condition $[a, JbJ^{-1}]=0$ for all
   $a,b\in\A$ is equivalent to
$[R, \bar M ]= 0$ and $[N, \bar Q]=0$. By \eqref{eq:QMexplicit} and
\eqref{eq:defQM}, one gets (omitting the indices $\dot s\dot t$ on
which all actions are trivial)
\begin{equation}
  \label{eq:56}
  [R, M]=
  \begin{pmatrix}
    [\delta_I^J R_r, M_r]&0 \\ 0& [\delta_I^J R_l, M_l].
  \end{pmatrix}_s^t.
\end{equation}
By \eqref{eq:defM}, one has
\begin{equation}
  \label{eq:62}
    [\delta_I^J R_r, M_r]=
    \begin{pmatrix}
      [\delta_I^J \sf d, \sf m \otimes \mathbb I_2] &0 \\ 0& [\delta_I^J p' ,
      \sf m'\otimes\mathbb I_2]
    \end{pmatrix}_\alpha^\beta,
\end{equation}
which is zero, as can be seen writing $\delta_I^J\sf d={\mathbb I}_4 \otimes \sf
d$ and similarly for $[\delta_I^J p' ,
      \sf m'\otimes\mathbb I_2] $. The same holds true for $[\delta_I^J R_l, M_l]$.
 \end{proof}
\subsection{Twisted fluctuation}
\label{subsec:twistfluc}

In the twisted context, fluctuations are similar to \eqref{eq/fluct}, replacing
the commutator for a twisted one \cite{Landi:2017aa}. In
  addition, if the twisted
first-order condition does not hold, one should add a non linear
term \cite{Chamseddine:2013fk,Martinetti:2021aa}. We thus consider the
\emph{twisted-covariant} Dirac operator \begin{equation}
D_{A}=D+A_{(1)} + \hat A_{(1)}  +A_{(2)}
\label{eq:12bis}
\end{equation}
where
\begin{equation}
A_{(1)}=\sum_i a_i\left[D,b_i\right]_\rho,\qquad a_i, b_i\in\A
\label{eq:44}
\end{equation}
is a twisted (generalised) $1$-form, $\hat A_{(1)}:= J A_{(1)} J^{-1}$
is its image by the conjugation with the real structure, while
\begin{equation}
A_{(2)}=\sum_i\hat {a}_i\left[A_\rho,\hat{b}_i\right]_\opr\quad
\text{ with }\; \hat{a}_i:=Ja_iJ^{-1} =\left(a_i^*\right)^\circ,\ \hat{b}_i:=JB_iJ^{-1}=\left(b^*_i\right)^\circ
\end{equation}
and $\rho^\circ$ denotes the automorphism of the opposite algebra defined
as
\begin{equation}
\label{eq:16}
  \rho^\circ(a^\circ):= (\rho^{-1}(a))^\circ.
\end{equation}
The term $A_{(2)}$ breaks the linearity of the map $A_{(1)}\to D + A_{(1)} +
J A_{(1)} J^{-1}$ and vanishes when the twisted first-order
condition \eqref{eq:73} holds (this is a straightforward adaptation to
the twisted context of the result of \cite{Chamseddine:2013fk}). We need to take it into account for, as explained in \S \ref{sec:twistfirstorder}, the twisted first-order condition only
  holds partially.

 The twisted $1$-form decomposes as the sum
$A_{(1)} = A_F + \slashed A$ of two pieces: one that we call the
\emph{finite part} of the fluctuation because it comes from the finite dimensional spectral triple, namely  
\begin{equation}
\label{defineAB}
  A_F=\sum_i a_i\left[\gamma_\mathcal{M}\otimes D_F,b_i\right]_\rho\quad a_i,
  b_i\in \A;
\end{equation}
another one coming from the manifold part of the spectral triple  
\begin{equation}
  \label{eq:40}
  \slashed{A}=\sum_i a_i\left[\slashed{D},b_i\right]_\rho\quad a_i,
  b_i\in \A
\end{equation}
that we call \emph{gauge part} in the following (terminology will
become clear later).

To guarantee that the twisted covariant operator \eqref{eq:12bis} is
selfadjoint, one assumes that the twisted $1$-form $A_{(1)}$ is
selfadjoint \cite[Prop.3.8]{Martinetti:2021aa}
(actually this is not a necessary condition, but requiring $A_{(1)}$ to
be selfadjoint makes sense viewing the fluctuation $D\to D_A$ as a
three steps process
\begin{equation}
D\rightarrow D+ A_{(1)}\rightarrow
D + A_{(1)} +\hat A_{(1)} \rightarrow D_A
\label{eq:26}
\end{equation}
such that selfadjointness
is preserved at each step). This means that for physical models,
we assume that both the gauge  $\slashed A$ and the finite $A_F$ parts
are selfadjoint.

So far, the construction works for any even dimension manifold
$\M$. To build explicitly the Standard Model, from now on one fixes the dimension of $\M$ to $m=4$. The grading and
the real structure are 
\begin{equation}\gamma_\M=\gamma^5= \gamma^0_E \gamma^1_E
\gamma^2_E \gamma^3_E=\begin{pmatrix}
    \I_4 &0 \\ 0 & -\I_4
  \end{pmatrix}_s^t=
\eta_s^t\delta_{\dot{s}}^{\dot{t}}
\label{eq:21}
\end{equation}
and
  \begin{equation}
\mathcal{J}= i\gamma^0_E\,\gamma^2_E\,  cc =i\left(
\begin{array}{cc} 
  {{\tilde\sigma}^2}&0_2\\ 
  0_2 & {\sigma^2}
\end{array}\right)_{s t} cc = -i\eta_s^t \tau_{\dot s}^{\dot t}
\, cc, 
\label{eq:2bbis}
\end{equation}
where $cc$ denotes the complex conjugation and we define
\begin{equation}
  \label{eq:37}
  \tau_{\dot s}^{\dot t}\coloneqq 
  \begin{pmatrix}
    0&-1\\1&0
  \end{pmatrix}_{\dot s}^{\dot t}, \quad  \eta_s^t \coloneqq   \begin{pmatrix}
    1 & \\ &-1
  \end{pmatrix}_s^t.
\end{equation}
For the internal spectral triple, one has
\begin{equation}
  \label{eq:18}
 \gamma_F=
\begin{pmatrix}
   \I_8 & & &\\  
& -\I_8 &  &\\
 &&  -\I_8 &\\  
&&& \I_8  \end{pmatrix}=\eta_{C\alpha}^{D\beta}\delta_I^J,\qquad J_{F}=\left(\begin{array}{cc}
0 & \mathbb{I}_{16}\\
\mathbb{I}_{16} & 0
\end{array}\right)_C^D cc= \xi_C^D\,\delta_{I\alpha}^{J\beta}
\end{equation}
where the matrix $\gamma_F$ is written in the basis left/right
particles then left/right antiparticles, and we define
\begin{equation}
  \label{eq:17}
\eta_\alpha^\beta \coloneqq   \begin{pmatrix}
    \I_2 & \\ &-\I_2
  \end{pmatrix}_\alpha^\beta,\quad \eta_C^D \coloneqq   \begin{pmatrix}
    1 & \\ &-1
  \end{pmatrix}_C^D,\quad \xi_C^D\coloneqq 
  \begin{pmatrix}
    0&1\\1&0
  \end{pmatrix}_C^D
\end{equation}
with $\eta_{C\alpha}^{D\beta}$ holding for $\eta_C^D\eta_\alpha^\beta$.  Thus 
\begin{equation}
  \Gamma =\gamma_\M\otimes \gamma_F=
  \eta_{sC\alpha}^{tD\beta}\delta_{\dot{s}}^{\dot{t}}\quad\text{ and } \quad
J = J_\M\otimes J_F =- i \eta^t_s\,\tau^{\dot t}_{\dot s}\, \xi^C_D \, \delta^{J\beta}_{I\alpha}\,cc.
\label{eq:12}
\end{equation}
\section{Scalar part of the twisted fluctuation}
\label{sec:scalarsector}
\setcounter{equation}{0} 

The scalar sector of the twisted Standard Model is obtained from the finite
part  \eqref{defineAB} of the twisted $1$-form, which in turns decomposes into a diagonal
part (determined by the Yukawa couplings of fermions)
\begin{equation}
  \label{eq:45}
  A_Y = \sum_i a_i[\gamma^5 \otimes D_Y, b_i]_\rho,
\end{equation}
and an off-diagonal part (determined by the Majorana mass of the neutrino)
\begin{equation}
  \label{eq:45bis}
  A_M = \sum_i a_i[\gamma^5 \otimes D_M, b_i]_\rho.
\end{equation}
As shown below, the former produces the Higgs sector, the
  latter a pair of extra scalar fields.
 
\subsection{The Higgs sector}
\label{subse:higgssector}

We begin with the diagonal  part \eqref{eq:45}. We first notice that the $M_3(\C)$ part of
the algebra \eqref{eq:3} twist-commutes with $\gamma^5\otimes D_Y$.
\begin{lemma}
\label{lem:diag1form} For any $b\in\A$ as in \eqref{eq:defb}, one has
 \begin{equation}
  \label{eq:46}
  \left[\gamma^5\otimes D_Y,b\right]_\rho=
  \begin{pmatrix}
S&0 \\ 0& 0
  \end{pmatrix}_D^C
\end{equation}
where $S$ has components
 \begin{align}
   \label{eq:47}
S_{s\dot s I\alpha}^{t\dot t J \beta} =\delta_{\dot s}^{\dot t}\left(\eta_s^u(D_0)_{I\alpha}^{J\gamma} \, 
  R_{s\gamma}^{t\beta}- 
  \rho(R)_{s\alpha}^{u\gamma}\, \eta_u^t 
 (D_0)_{I\gamma}^{J\beta}\right).
 \end{align}
\end{lemma} 
\begin{proof}From the explicit forms \eqref{eq:9} of $D_Y$ and
  \eqref{eq:defb} of $b$, one has 
\begin{align*}
 \left[\gamma^5\otimes D_Y,b\right]_\rho&=
\left(\begin{matrix}
\left[\gamma^5\otimes D_0,R\right]_\rho&\\&\left[\gamma^5\otimes D_0^\dagger,N\right]_\rho
\end{matrix}\right)_C^D.
\end{align*}
In the tensorial notation, 
$S\coloneqq [\gamma^5\otimes D_0, R]_\rho$ has components
\begin{align}
  \label{eq:48}
S_{s\dot s I\alpha}^{t\dot t J \beta} &=\eta_s^u \delta_{\dot s}^{\dot
  u}(D_0)_{I\alpha}^{K\gamma} \, \delta_{\dot uK}^{\dot t J}
  R_{u\gamma}^{t\beta}- \delta_{\dot sI}^{\dot u K}
  \rho(R)_{s\alpha}^{u\gamma}\, \eta_u^t \delta_{\dot u}^{\dot
  t}(D_0)_{K\gamma}^{J\beta},\\
&= \delta_{\dot s}^{\do t}\left(\eta_s^u(D_0)_{I\alpha}^{J\gamma} \, 
  R_{u\gamma}^{t\beta}- 
  \rho(R)_{s\alpha}^{u\gamma}\, \eta_u^t 
 (D_0)_{I\gamma}^{J\beta}\right),
\end{align}
which shows \eqref{eq:47}.
To show that
\begin{equation}
 [\gamma^5\otimes D_0^\dag, N]_\rho=0,
\label{eq:65}
 \end{equation}
let us denote $T$ the left-hand side of the equation above. It  has components
\begin{align}
T_{s\dot s I\alpha}^{t\dot t J \beta} &=\eta_s^u \delta_{\dot s}^{\dot
  u}(D_0^\dag)_{I\alpha}^{K\gamma} \, \delta_{\dot u}^{\dot t}
  N_{u\gamma K}^{t\beta J}- \delta_{\dot s}^{\dot u}
\,\rho(N)_{s\alpha I}^{u\gamma K}\, \eta_u^t \delta_{\dot u}^{\dot t}(D_0^\dag)_{K\gamma}^{J\beta},\\
&= \delta_{\dot s}^{\dot t}\left(\eta_s^u(D_0^\dag)_{I\alpha}^{K\gamma} \, 
  N_{u\gamma K}^{t\beta J}- 
  \rho(N)_{s\alpha I}^{u \gamma K}\, \eta_u^t 
 (D_0^\dag)_{K\gamma}^{J\beta}\right),\\
\label{eq:4.64}
&= \delta_{\dot s}^{\dot t}
  \begin{pmatrix}
    (D_0^\dag)_{I\alpha}^{K\gamma} \, (N_r)_{\gamma K}^{\beta J}-
    (N_l)_{\alpha I}^{\gamma K}\,(D_0^\dag)_{K\alpha}^{J\gamma} &0\\
0& - (D_0^\dag)_{I\alpha}^{K\gamma} \, (N_l)_{\gamma K}^{\beta J}+
    (N_r)_{\alpha I}^{\gamma K}\,(D_0^\dag)_{K\gamma}^{J\beta}
  \end{pmatrix}_s^t.
\end{align}
Since $(D_0^\dagger)_I^K =\delta_I^K(D_0^I)$ and $(D_0^\dagger)^K_J =\delta_J^K(D_0^J)$ (with no summation
on $I$ and $J$), the upper-left term in \eqref{eq:4.64} is
\begin{equation}
  \label{eq:69}
\left(D_0^I\right)_\gamma^\delta (N_r)_{\delta I}^{\beta J}-{(N_l)}_{\alpha I}^{\gamma J} \left(D_0^{J}\right)_\gamma^\beta=
  \begin{pmatrix}
0&   \bar{\sf k}^I(\sf n'\otimes \mathbb I_2) \\\ {\sf k}^I(\sf n
  \otimes \mathbb I_2)&0
  \end{pmatrix}_\alpha^\beta -  
\begin{pmatrix}
0&  (\sf n'\otimes \mathbb I_2)  \bar{\sf k}^J\\ 
(\sf n \otimes \mathbb I_2)  {\sf k}^J&0
  \end{pmatrix} _\alpha^\beta
\end{equation}
where we omitted the $I, J$ indices on $\sf n$. One has
\begin{equation}
{\sf k}^I({\sf n}
  \otimes \mathbb I_2)=  \begin{pmatrix}
  \label{eq:70}
k_u^I {\sf n} & \\ & k_d^I {\sf n}
     \end{pmatrix}, \quad ({\sf n}
  \otimes \mathbb I_2) \sf k^J=  \begin{pmatrix}
{\sf n} k_u^J  & \\ & {\sf n} k_d^J 
     \end{pmatrix}, 
\end{equation}
and similarly for the terms in $\sf n'$. Restoring the indices,  one has
\begin{equation}
  \label{eq:71}
 k_u^I {\sf n}_I^J=
  \begin{pmatrix}
    k_u^l d &\\ &k_u^q n
  \end{pmatrix}_I^J, \quad  {\sf n}_I^J k_u^J=
  \begin{pmatrix}
    d k_u^l  &\\ & n k_u^q
  \end{pmatrix}_I^J
\end{equation}
where we write $k_u^{I=0}=k_u^l$ for the lepton, and $k_u^{I=1,2,3}=k^q_u$
for the coloured quarks. Again,  in the expression above, there is no
summation on $I$ and $J$: $k_u^I {\sf n}_I^J$ means the matrix $\sf n$ in which the
$I^\text{th}$ line is multiplied by $k_u^I$, while in ${\sf n}_I^J k_u^J$
this is
the $J^\text{th}$ column of $\sf n$ which is multiplied by $k_u^J$.
Therefore
\begin{equation}
{\sf k}^I({\sf n}
  \otimes \mathbb I_2) - ({\sf n}
  \otimes \mathbb I_2) {\sf k}^J =0.
\label{eq:72}
  \end{equation}

Similarly $\bar{\sf k}^I({\sf n'}
  \otimes \mathbb I_2) - ({\sf n'}
  \otimes \mathbb I_2) \bar{{\sf k}^J} =0$, so that \eqref{eq:69} -
  that is the upper-left term in  \eqref{eq:4.64} - is zero. The
  proof that the lower-right term is zero is similar. Hence
  \eqref{eq:65} and the result.\end{proof}

A similar result holds in the non-twisted case (the
  computation is similar as above, with $\sf n' = \sf n$, so that 
  everything boils down to the single equation \eqref{eq:72}). The
  result however is not true if one genuinely generalises the twist
  used in
\cite{buckley}. As explained below, this yields an
additional  violation of the twisted first order
condition, besides the one required to generate the 
field $\sigma$. That is why we do not use this genuine twist, but
rather the one  presented in~section~\ref{sec:minimaltwist}.

\begin{remark}
 The twist in \cite{buckley}  was not applied to the $M_3(\C)$ part of
the algebra. Only $\C \oplus \HHH$ was doubled
and this yielded an action similar as  the one used on the present
paper (modulo a change of notations, the representation (4.7) of
\cite{buckley} coincides with \eqref{eq:defQM}).   A genuine generalisation of
this twist consists
in making two copies of $M_3(\C)$ acting independently on the left and
right components of spinors, namely 
$a\in\mathcal{A}$ acts as in \eqref{eq:QMexplicit}, but now $M_{r,l}$ are given by
\begin{equation}
M_r=\left(\begin{matrix}
{\sf m}\otimes\mathbb{I}_4
\end{matrix}\right)_\alpha^\beta,\qquad M_l=\left(\begin{matrix}
{\sf m}'\otimes\mathbb{I}_4
\end{matrix}\right)_\alpha^\beta.
\end{equation}
Then lemma \ref{lem:diag1form} no longer holds for the lower right
term $T$ is non necessarily zero (on the r.h.s. of \eqref{eq:69} the first parenthesis now contains only $\sf n$,
and the second only $\sf n'$, so that the cancellation
\eqref{eq:72} is no longer true). 
\end{remark}
We now compute the $1$-forms generated by the Yukawa couplings of the
fermions. In order to do so, we extend the action of the automorphism $\rho$ to any polynomial in
$q, q'$, $p,p'$, $c, c'$, $d, d'$. Namely $\rho$ ``primes'' what is
un-primed, and vice-versa.  For instance $\rho(qp' -c'd)= q'p - cd'$.

\begin{proposition}
\label{prop:diag1form} The diagonal part (\ref{eq:45}) of a twisted $1$-form
 is   \begin{equation}
A_Y=\left(\begin{matrix}
A&\\ & 0
\end{matrix}\right)_C^D
\quad\text{ where } \quad  A=\delta_{\dot s I}^{\dot t J}
  \begin{pmatrix}
    A_r & \\ &A_l
  \end{pmatrix}_s^t
\label{eq:49}
\end{equation}
with
\begin{equation}
  \label{eq:67}
 A_r=
  \begin{pmatrix}
    & \bar{\sf k^I}H_1\\ H_2\sf k^I
  \end{pmatrix}_\alpha^\beta,\quad A_l=-
  \begin{pmatrix}
    & \bar{\sf k^I}H'_1\\ H'_2\sf k^I
  \end{pmatrix}_\alpha^\beta,\;
\end{equation}
where $H_{i=1, 2}$ and $H'_{1,2}=\rho(H_{1,2})$ are quaternionic
fields.
\end{proposition}
\begin{proof}
From \eqref{eq:27} and lemma \ref{lem:diag1form}, one
has $a[\gamma^5\otimes D_Y, b]_\rho= Q\, S$. In components, this gives (using  the explicit forms
\eqref{eq:QMexplicit} of $Q$, $R$):
\begin{align}
A_{s\dot{s}I\alpha}^{t\dot{t}J\beta}&=Q_{s\dot{s}I\alpha}^{u\dot{u}K\gamma}\delta^{\dot t}_{\dot u}\left[\eta_u^v\,(D_0)_{K\gamma}^{J\delta}\,
                                        R_{v\delta}^{t\beta}-\rho
                                        (R)_{u\gamma}^{v\delta}\,\eta_v^t
                                       \, (D_0)_{K\delta}^{J\beta}\right]
\\
&=\delta_{\dot s I}^{\dot t J} Q_{s\alpha}^{u\gamma}\left[\eta_u^v\,\left(D_0^I\right)_{\gamma}^{\delta}\,R_{v\delta}^{t\beta}-\rho(R)_{u\gamma}^{v\delta}\,\eta_v^t\,\left(D_0^I\right)_{\delta}^{\beta}\right],
\end{align}
where we use $\delta_I^K (D_0)_K^J= \delta_I^J(D_0^I)$ (with no
summation on $I$ in the last expression). 
 Since $Q$ is diagonal on the chiral
indices $s$, the only non-zero
components of $A$ are for $s=t=r$ and $s=t=l$, namely
\begin{align}
\label{AIR0}
A^{r\dot tJ\beta}_{r\dot s I \alpha} = \delta_{\dot sI}^{\dot
  tJ}\left(A_r^I\right)_\alpha^\beta &\quad \text{ with }\;
                                       \left(A_r^I\right)_\alpha^\beta
                                       =\left(Q_r\right)_{\alpha}^{\gamma}\left[\left(D_0^I\right)_{\gamma}^{\delta}\left(R_r\right)_{\delta}^{\beta}-\left(R_l\right)_{\gamma}^{\delta}\left(D_0^I\right)_{\delta}^{\beta}\right],\\
\label{AIL0}
A^{l\dot tJ\beta}_{l\dot s I \alpha} = \delta_{\dot sI}^{\dot
  tJ}\left(A_l^I\right)_\alpha^\beta &\quad \text{ with }\; \left(A_l^I\right)_\alpha^\beta=\left(Q_l\right)_{\alpha}^{\gamma}\left[-\left(D_0^I\right)_{\gamma}^{\delta}\left(R_l\right)_{\delta}^{\beta}+\left(R_r\right)_{\gamma}^{\delta}\left(D_0^I\right)_{\delta}^{\beta}\right].
\end{align}
From the explicit expression \eqref{eq:defQ},  \eqref{eq:31},
\eqref{eq:defD0I} of $Q_{r/l}$, $R_{r/l}$ and $D_0^I$ one gets
\begin{align}
  \label{eq:63}
  Q_r D_0^IR_r=
  \begin{pmatrix}
    & {\sf c }\bar{\sf k^I} p' \\  q'{\sf k^I}\sf d 
  \end{pmatrix}_\alpha^\beta,\; 
Q_r R_l  D_0^I=
  \begin{pmatrix}
    & {\sf c } {\sf d'}\bar{\sf k^I}  \\  q' p{\sf k^I}
  \end{pmatrix}_\alpha^\beta,\\
  Q_l D_0^IR_l=
  \begin{pmatrix}
    & {\sf c' }\bar{\sf k^I} p \\  q{\sf k^I}\sf d' 
  \end{pmatrix}_\alpha^\beta,\; 
Q_l R_r  D_0^I=
  \begin{pmatrix}
    & {\sf c'}{\sf d}\bar{\sf k^I}  \\  q p '{\sf k^I}
  \end{pmatrix}_\alpha^\beta,\end{align}
Using that $\sf c, \sf c', \sf d, \sf d'$ commute with ${\sf k^I}$, one
has
\begin{align}
  \label{eq:64}
  Q_r(D_0^IR_r - R_lD_0^I) =
  \begin{pmatrix}
 & \bar{\sf k}^I H_1\\ H_2 {\sf k}^I&    
  \end{pmatrix}_\alpha^\beta,\quad  -Q_l D_0^IR_l + Q_lR_rD_0^I=
  -\begin{pmatrix}
    & \bar{\sf k}^I H_1'\\ H_2'{\sf k}^I
  \end{pmatrix}_\alpha^\beta
\end{align}
where
\begin{align}
  \label{eq:65bis}
  H_1\coloneqq {\sf c}(p'-{\sf d'}), \quad H_2\coloneqq q'({\sf d}-p),\quad
  H'_1\coloneqq {\sf c'}(p-{\sf d}), \quad H'_2\coloneqq q({\sf d'}-p').
\end{align}
This shows the result.
\end{proof}

Imposing now selfadjointness as stressed before \eqref{eq:26}at the
  beginning of this section, we get the
\begin{corollary}
\label{cor:higgs}
A selfadjoint diagonal twisted $1$-form \eqref{eq:45} is parametrized
by two independent 
scalar  quaternionic field $H_r, H_l$.
\end{corollary}
\begin{proof}
The twisted $1$-form \eqref{eq:49} is selfadjoint if and only if
\begin{equation}
H_2=H_1^\dag\eqqcolon H_r\quad\text{ and } \quad {H'_2} ={H'}^\dag_1\eqqcolon H_l.\label{eq:75}
\end{equation}
They are independent as
follows from their definition \eqref{eq:65bis}. 
\end{proof}

Since $\gamma^5\otimes D_Y$ satisfies the twisted first-order
  condition (Prop.~\ref{prop:twisted-first-order}), it does  not contribute to the non-linear term $A_{(2)}$
of the twisted fluctuation. Gathering the results of this section,
one thus works out the fields induced by the
Yukawa coupling of fermions  via a
twisted fluctuation of the metric.
\begin{proposition}
\label{prop:diagfluc}
A selfadjoint diagonal fluctuation is
\begin{equation}
D_{A_Y}=\gamma^5\otimes D_Y+A_Y+ \widehat{A_Y}=
\begin{pmatrix}
 \eta_s^t\delta_{\dot s}^{\dot t} D_0 + A\\ &  \eta_s^t\delta_{\dot
   s}^{\dot t} D_0^\dag +\bar A
\end{pmatrix}_C^D
\end{equation}
where $A=\delta_{\dot s I}^{\dot t J}
\begin{pmatrix}
  A_r& \\ & A_l
\end{pmatrix}_s^t$
is generated by two quaternionic fields $H_r, H_l$ as 
\begin{equation}
  \label{eq:66}
  A_r=
  \begin{pmatrix}
    & \bar{\sf k}^I H_r^\dag\\ H_r{\sf k}^I
  \end{pmatrix}_\alpha^\beta,   \quad A_l=
  \begin{pmatrix}
    & \bar{\sf k}^I H_l^\dag\\ H_l{\sf k}^I
  \end{pmatrix}_\alpha^\beta.
\end{equation}
\end{proposition}
\begin{proof}
Remembering that $J^{-1}=-J$, proposition \ref{prop:diag1form} yields 
\begin{align}
\label{eq:jayj}
\widehat{A_Y} = JA_Y J^{-1}&=   \begin{pmatrix} 0& \cal J \\ {\cal J} &0 \end{pmatrix}_C^D
\begin{pmatrix}
A& \\ & 0
\end{pmatrix}_C^D 
 \begin{pmatrix}
0& -\cal J \\ -{\cal J} &0       
\end{pmatrix}_C^D= \begin{pmatrix}
0& \\ & -{\cal J}A {\cal J}^{-1}      
\end{pmatrix}_C^D.
\end{align}
From the explicit form \eqref{eq:2bbis} of ${\cal J}=-{\cal J}$ and \eqref{eq:49} of
$A$,  one obtains (omitting the $I J$ and  $\alpha\beta$ indices in which the
real structure $J$ is
trivial)
\begin{align}
  \label{eq:58}
  {\cal J} A {\cal J}^{-1} &=\eta_s^u\tau_{\dot s}^{\dot u}
  \,\bar A_{u\dot u}^{v\dot v}\, \eta_v^t\tau_{\dot v}^{\dot t}
= \eta_s^u\tau_{\dot s}^{\dot u}
  \begin{pmatrix}
   \bar A_r&0 \\ 0& \bar A_l
  \end{pmatrix}_s^t
\eta_v^t\tau_{\dot v}^{\dot t}=  \begin{pmatrix}
    -\bar A_r&0 \\ 0& -\bar A_l
  \end{pmatrix}_s^t=-\bar A,
\end{align}
where we used \eqref{eq:49}  and write $\tau_{\dot s}^{\dot u}\delta_{\dot u}^{\dot v}\tau_{\dot v}^{\dot t}=-\delta_{\dot s}^{\dot t}$.

The result follows summing \eqref{eq:jayj} with $A_Y$ given in
Prop. \ref{prop:diag1form}  and $D_Y$ given in \eqref{eq:9}, then
using corollary \ref{cor:higgs} to rename $H_r$ and $H_l$.
\end{proof}

In the non-twisted case, the primed and unprimed quantities are
  equal, so that one obtains only one quaternionic field $H_r=H_l$,
  which combines in the action as
  \begin{equation}
\label{eq:Higgsrl}
H\coloneqq H_r +
  H_l =
  \begin{pmatrix}
    \phi_1 & -\bar\phi_2\\ \phi_2 & \bar\phi_1
  \end{pmatrix},
\end{equation}
whose complex components $\phi_1, \phi_2$ identify with the Higgs
doublet. In the twisted case,  the complex components $\phi^r_{1,2}$,
$\phi^l_{1,2}$ of $H_r$, $H_l$  define two scalar doublets
  \begin{equation}
\label{eq:Higgscomp}
    \Phi_r \coloneqq 
    \begin{pmatrix}
      \phi^r_1\\ \phi^r_2
    \end{pmatrix},\;  \Phi_l \coloneqq 
    \begin{pmatrix}
      \phi^l_1\\ \phi^l_2
    \end{pmatrix},
  \end{equation}
which act respectively on the right and on the left part of
  the Dirac spinors. However, similar to \eqref{eq:Higgsrl} they only appear in the fermionic action through their linear
  combination $H_r+H_l$ \cite{Manuel-Filaci:2020aa}, therefore there is actually only one physical Higgs doublet in the twisted case as well.
 
\subsection{The extra scalar field}
\label{subsec:sigma}

The computation of the off-diagonal term \eqref{eq:45bis} of the
finite part of the twisted
$1$-form is easier than for the diagonal part,
because  $D_M$ has only one non-zero component.
\begin{proposition} The off-diagonal part \eqref{eq:45bis} of a
  twisted $1$-form is
  \begin{equation}
    \label{eq:68}
    A_M=
    \begin{pmatrix}
& C\\ D&      
    \end{pmatrix}_C^D
  \end{equation}
  where
  \begin{equation}
    \label{eq:81}
    C=k_R\, \delta_{\dot s}^{\dot t}
     \begin{pmatrix}
        C_r & \\ & C_l
      \end{pmatrix}_s^t,\quad   D=\bar  k_R\, \delta_{\dot s}^{\dot t}
     \begin{pmatrix}
        D_r & \\ & D_l
      \end{pmatrix}_s^t
   \end{equation}
with 
\begin{equation}
\label{81bbis}
C_r=D_r=\Xi_{I\alpha}^{J\beta}\sigma,\quad C_l=D_l= -\Xi_{I\alpha}^{J\beta}\sigma' 
  \end{equation}
where $\sigma$ and $\sigma'$ are complex  fields.
\end{proposition}
\begin{proof}
Using the explicit form \eqref{eq:9}  of $D_M$, for $a$ in
\eqref{eq:27} and $b$ in \eqref{eq:defb} one gets
\begin{align}
\nonumber
 a\left[\gamma^5\otimes D_M,b\right]_\rho &= 
\left(\begin{matrix}
Q&0\\0&M
\end{matrix}\right)
\left[\left(\begin{matrix}
0&\gamma^5\otimes D_R\\
\gamma^5\otimes D_R^\dagger&0
\end{matrix}\right),\left(\begin{matrix}
R&0\\0&N
\end{matrix}\right)\right]_\rho=\\
&=\left(\begin{matrix}
\!\!\!\!&\!\!\!\!\!\!\!\!\!\!\!\!\!\! Q\left(\left(\gamma^5\otimes D_R\right) \!N\!\!-\!\!\rho\left(R\right)\!\left(\gamma^5\otimes D_R\right)\right) \\
 M\!\left((\gamma^5\otimes D_R^\dagger) R\!\!-\!\!\rho\left(N\right)\!(\gamma^5\otimes D_R^\dagger)\right)&\!\!\!\!
\end{matrix}\right)_C^D\!\!.
\label{eq:comMajo}
\end{align}

With  $D_R$ given in \eqref{eq:Drtensor}, one computes
the upper-right component $C$ of the matrix above:
\begin{equation}
C_{s\dot{s}I\alpha}^{t\dot{t}J\beta}=Q_{s\dot{s}I\alpha}^{u\dot{u}K\gamma}\left[k_R\,\eta_u^v\delta_{\dot{u}}^{\dot{v}}\, \Xi_{K\gamma}^{L\delta} N_{v\dot{v}L\delta}^{t\dot{t}J\beta}-k_R\,\rho(R)_{u\dot{u}K\gamma}^{v\dot{v}L\delta}\eta_v^t\delta_{\dot{v}}^{\dot{t}} \,\Xi_{L\delta}^{J^\beta}\right].
\end{equation}
Since $Q, N$ are diagonal in the $s$ index and proportional
to
$\delta_{\dot{s}}^{\dot{t}}$, the non-zero components of $C$~are
\begin{align}
\label{eq:77}
(C_r)_{I\alpha}^{J\beta}=k_R\,\delta_{\dot s}^{\dot t}\,
  (Q_r)_{I\alpha}^{K\gamma}
\left[\Xi_{K\gamma}^{L\delta}
  (N_r)_{L\delta}^{J\beta}-(R_l)_{K\gamma}^{L\delta}
  \,\Xi_{L\delta}^{J^\beta}\right],\\
(C_l)_{I\alpha}^{J\beta}=k_R\,\delta_{\dot s}^{\dot t}\,
  (Q_l)_{I\alpha}^{K\gamma}
\left[-\Xi_{K\gamma}^{L\delta} (N_l)_{L\delta}^{J\beta}+(R_r)_{K\gamma}^{L\delta} \,\Xi_{L\delta}^{J^\beta}\right].
\end{align}

Explicitly, from the formula \eqref{eq:defQ} for $Q_{r/l}$  and
\eqref{eq:31} of $R_{r/l}, N_{r/l}$, one gets
\begin{align*}
  \label{eq:78}
  Q_r(\Xi N_r-R_l\Xi)&=
  \begin{pmatrix}
   {\sf  c}\delta_I^J \!\!  & \\ &\!\!  q' \delta_I^j
  \end{pmatrix}_\alpha^\beta\left(  \begin{pmatrix}
    \Xi_I^J   & \\ & 0_3
  \end{pmatrix}_\alpha^\beta   \begin{pmatrix}
    {\sf n}\otimes \mathbb I_2   & \\ & {\sf n'}\otimes \mathbb I_2
  \end{pmatrix}_\alpha^\beta- \begin{pmatrix}
    {\sf d'}\delta_I^J   & \\  & p\,\delta_I^J 
  \end{pmatrix}_\alpha^\beta \begin{pmatrix}
    \Xi_I^J   & \\ & 0_3
  \end{pmatrix}_\alpha^\beta   \right)=\\
&=
  \begin{pmatrix}
   {\sf  c}\delta_I^J  \!\! & \\\!\! & q' \delta_I^j
  \end{pmatrix}_\alpha^\beta \begin{pmatrix}
    \Xi_I^J d- d' \Xi_I^J & \\ & \!\!\!\!\! 0_3
  \end{pmatrix}_\alpha^\beta   = 
                                             \begin{pmatrix}
  c ( d- d') \Xi_I^J  & \\ & \!\!\!\!\! 0_3 
  \end{pmatrix}_\alpha^\beta=\sigma \Xi_{\alpha I}^{J\beta} 
\end{align*} 
and similarly
\begin{equation}
  \label{eq:80}
   Q_l(-\Xi N_l +R_r\Xi)=-\sigma' \Xi_{\alpha I}^{J\beta} 
\end{equation}
where we define the scalar fields
\begin{equation}
  \label{eq:79}
  \sigma \coloneqq c ( d- d'),\quad   \sigma' \coloneqq  c' ( d'- d).
\end{equation}

Similarly, one computes that 
 the lower left component $D$ of \eqref{eq:comMajo} has non zero
 components 
 \begin{align}
   \label{eq:83}
   D_r &= \bar k^R \delta_{\dot s}^{\dot t} \, M_r(\Xi R_r - N_l\Xi) =
   \bar k^R\, \delta_{\dot s}^{\dot t}\, \Xi_{\alpha I}^{\beta J}c(d
   -d')=  \bar k^R\, \delta_{\dot s}^{\dot t}\, \Xi_{\alpha I}^{\beta
   J}\, \sigma,\\
\nonumber
D_l &= \bar k^R \delta_{\dot s}^{\dot t} \, M_l(-\Xi R_l + N_r\Xi) =
   \bar k^R\, \delta_{\dot s}^{\dot t}\, \Xi_{\alpha I}^{\beta J}c'(-d'
   +d)=  -\bar k^R\, \delta_{\dot s}^{\dot t}\, \Xi_{\alpha I}^{\beta
   J}\, \sigma'.
 \end{align}

\vspace{-.7truecm}
\end{proof}
\noindent An off-diagonal  $1$-form $A_M$ is 
self-adjoint if and only if
$D_r^\dag =C_r$ and $D_l^\dag=C_l$, that is
\begin{equation}
  \label{eq:87}
  \sigma =\bar\sigma,\quad  \sigma' =\bar\sigma'. 
\end{equation}

The part of the twisted fluctuation induced by the Majorana mass of
the neutrino is then easily
obtained, taking however into account the 
  contribution of $D_M$ to the non-linear term $A_{(2)}$, since
  $\gamma^5\otimes D_M$ violates the twisted first-order condition (cf.
Prop.~\ref{prop:twisted-first-order}).
\begin{proposition}
\label{proposi:offdiagfluc}
  A off-diagonal fluctuation is parametrised by  two
    independent real scalar fields
$\sigma_r, \sigma_l$:
\begin{equation}
\label{prop:offdiagfluc}
D_{A_M}=\gamma^5\otimes D_M+A_M+ \widehat{A_M}+ {{A_M}_{(2)}}=
\delta_{\dot t}^{\dot t} \left(\begin{matrix}
0&\eta_s^t D_{0}+k_R\,\Xi_{I\alpha}^{J\beta}\,\bar\Sigma_s^t\\ \eta_s^t
D_{0}^\dagger+\bar k_R \,\Xi_{I\alpha}^{J\beta}\,  \Sigma_s^t&0
\end{matrix}\right)_C^D.
\end{equation}
where
\begin{equation}
  \label{eq:SIGMA}
\Sigma=
\begin{pmatrix}
  \sigma_r & \\ & \sigma_l 
\end{pmatrix}_s^t.\end{equation}
\end{proposition}
\begin{proof}
 As in the proof of proposition \ref{prop:diagfluc},  one has
  \begin{equation}
    \label{eq:88}
   \widehat{A_M}=  J A_M J^{-1}=
    \begin{pmatrix}
     0&-{\cal J}D {\cal J}\\ -{\cal J}C {\cal J}&0
    \end{pmatrix}_C^D
  \end{equation}
with
\begin{equation}
  \label{eq:84}
  {\cal J}C {\cal J}^{-1}=\eta_s^u\tau_{\dot s}^{\dot u}
  \,\bar C_{u\dot u}^{v\dot v}\, \eta_v^t\tau_{\dot v}^{\dot t}= -\bar C
\end{equation}
and similarly for $D$. Hence
\begin{equation}
  \label{eq:85}
  A_M+ \widehat{JA_M}=
\left(\begin{matrix}
0&C+\overline{D}\\ \overline{C}+D&0
\end{matrix}\right)_C^D.
\end{equation}

The non linear term is 
 (omitting the summation index)
\begin{equation}
{A_M}_{(2)}=\hat {a}\left[A_M,\hat{b}\right]_\opr.
\end{equation}
By proposition \ref{prop:algopp} and the explicit form \eqref{eq:68}
of $A_M$ one gets
\begin{equation}
\hat a \left[A_M,\hat b\right]_\crho= -\begin{pmatrix}
      \bar M&0 \\ 0&\bar Q
    \end{pmatrix}_C^D\begin{pmatrix}
0&\overline{\rho(N)}C-C\overline{R}\\\overline{\rho(R)}D-D\overline{N}
\end{pmatrix}_C^D,
\end{equation}
where we use $\rho^\circ(\hat b)=\rho^\circ((b^*)^\circ)=(\rho^{-1}(b^*))^\circ=
(\rho(b)^*)^\circ=\hat{\rho(b)}$ which follows from the definition
\eqref{eq:16} of $\rho^\circ$ together with the regularity condition
$\rho(a^*)=(\rho^{-1}(a))^*$ satisfied by $\rho$.
From \eqref{eq:81} and \eqref{eq:defb}
\begin{align}
C\overline{R}&=k_R\delta_{\dot s}^{\dot t}\Xi_{I\alpha}^{J\beta}\begin{pmatrix}
\overline{d}\sigma\\&-\overline{d'}\sigma'
\end{pmatrix}_s^t,&
\overline{\rho(N)}C&=k_R\delta_{\dot s}^{\dot t}\Xi_{I\alpha}^{J\beta}\begin{pmatrix}
\overline{d'}\sigma\\&-\overline{d}\sigma'
\end{pmatrix}_s^t,\\
\overline{\rho(R)}D&=\overline{k_R}\delta_{\dot s}^{\dot t}\Xi_{I\alpha}^{J\beta}\begin{pmatrix}
\overline{d'}\sigma\\&-\overline{d}\sigma'
\end{pmatrix}_s^t,&
D\overline{N}&=\overline{k_R}\delta_{\dot s}^{\dot t}\Xi_{I\alpha}^{J\beta}\begin{pmatrix}
\overline{d}\sigma\\&-\overline{d'}\sigma'
\end{pmatrix}_s^t.
\end{align}
Remembering \eqref{eq:79}, one obtains
\begin{equation}
-\bar M\left(\overline{\rho(N)}C-C\overline{R}\right)=k_R\delta_{\dot s}^{\dot t}\Xi_{I\alpha}^{J\beta}\begin{pmatrix}
\bar c\left(\overline{d}-\overline{d}'\right)\sigma\\&\overline{c'}\left(\overline{d}-\overline{d}'\right)\sigma'
\end{pmatrix}_s^t=k_R\delta_{\dot s}^{\dot t}\Xi_{I\alpha}^{J\beta}\begin{pmatrix}
|\sigma|^2\\&-|\sigma'|^2
\end{pmatrix}_s^t,
\end{equation}
\begin{equation}
-\bar Q\left(\overline{R}'D-D\overline{N}\right)=\overline{k_R}\delta_{\dot s}^{\dot t}\Xi_{I\alpha}^{J\beta}\begin{pmatrix}
\bar c\left(\overline{d}-\overline{d}'\right)\sigma\\&\overline{c'}\left(\overline{d}-\overline{d}'\right)\sigma'
\end{pmatrix}_s^t=\begin{pmatrix}
|\sigma|^2\\&-|\sigma'|^2
\end{pmatrix}_s^t.
\end{equation}
Hence
\begin{equation}
\omg=\delta_{\dot s}^{\dot t}\Xi_{I\alpha}^{J\beta}\begin{pmatrix}
0&k_R\\\overline{k_R}&0
\end{pmatrix}_C^D\begin{pmatrix}
\abs{\sigma}^2\\&-\abs{\sigma'}^2
\end{pmatrix}_s^t.
\end{equation}

The explicit form of $\Sigma$ follows from \eqref{eq:81}-\eqref{81bbis}, defining
\begin{equation*}
\label{eq:19}
\sigma_r=
\bar\sigma + \sigma+ |\sigma|^2\; \text{ and }\; \sigma_l= -\bar\sigma' -\sigma'
- |\sigma'|^2.
\end{equation*}

\vspace{-.9truecm}\end{proof}

\noindent The non-linear term does not modify the nature of the extra-scalar
field $\sigma$. It simply modifies the relation between the components
$\sigma_r, \sigma_l$ and the elements of the algebra defining the
twisted $1$-form, introducing the terms $|\sigma|^2$, $|\sigma'|^2$
in the equation above.
\begin{remark}
\label{rem:extra-scalar-field}  The field $\sigma$ is chiral, in the sense it has two independent
  components $\sigma_r$, $\sigma_l$. The one initially worked out in
  \cite{buckley} was not chiral. This is because in the latter case,
 one does not double $M_3(\C)$ and identifies the complex component of
 $\sf m$ with the complex component of $Q_r$. This means that the
 component $d'$ of $N_l$ identifies with the component $d$ of $R_r$, so that
 \eqref{eq:80}  and \eqref{eq:83} vanish, that is $C_l = D_r=0$. Similarly,  the component  $c'$ of $M_l$
 becomes $c$, so that $D_l=C_r$. One thus retrieves the formula
 $(4.32)$ of \cite{buckley} (in which the role of $c$ and $d$ have
 been interchanged). However, forcing the identification of the
 (non-doubled) 
 $M_3(\C)$-component with one of the (doubled) component of $\C$ is
 actually not compatible with the twist, as explained in greater
 details in \cite{Manuel-Filaci:2020aa}. This problem is resolved in
 the present paper, where $M_3(\C)$ is doubled and there is a minimal
 violation of the twisted first-order condition.
\end{remark}
As an illustration that the selfadjointness of the $1$-form is not
necessary to get a selfadjoint twisted fluctuation (see \S ~\ref{subsec:twistfluc}), notice that in the proposition above $D_{A_M}$ is selfadjoint regardless of the
selfadjointness of $A_M$. As well, one does not need to assume that
$A_M$ is selfadjoint to ensure that the fields $\sigma_r, \sigma_l$
are real.

\section{Gauge part of the twisted fluctuation}
\label{sec:gauge}
\setcounter{equation}{0}

In this section, we compute the twisted fluctuation induced by the
free part 
$\slashed{D}= \slashed\partial \otimes \mathbb I_F$
 of the Dirac operator \eqref{eq:04}, that is
\begin{equation}
\slashed D +\slashed A + J\slashed A\, J^{-1}
 \label{eq:115}
\end{equation}
where $\slashed A$ is  the twisted
$1$-form \eqref{eq:40} induced by $\slashed D$, that we call in the following a
\emph{free $1$-form}. As will be checked in section  \ref{sec:gaugetransform}, the components
of this form are the gauge fields of the model.
There is no  non-linear term $\slashed A_{(2)}$, for $\slashed{D}$ does verify the twisted
  first-order condition, as shown in proposition \ref{prop:twisted-first-order}.
\subsection{Dirac matrices and twist}
\label{subsec:gamma}

We begin by recalling some useful
relations between the Dirac matrices and the twist. 

\begin{lemma}
If an operator $\cal O$ on $L^2(\M, S)$ twist-commutes with the Dirac
matrices,
\begin{equation}
\gamma^\mu\cal O = \rho(\cal O)\gamma^\mu \quad \forall 
\mu\label{eq:92}
\end{equation}
for some automorphism $\rho$ of ${\cal B}(\HH)$, and commutes the spin connection $\omega_\mu$, then 
 \begin{equation}
\left[\slashed{\partial},\cal
  O\right]_\rho=-i\gamma^\mu\partial_\mu \cal O.
\label{eq:89}
\end{equation}
\end{lemma}
\begin{proof}
One has
\begin{equation}
  \label{eq:96}
    [\gamma^ \mu\nabla_\mu, \cal O]_\rho=   [\gamma^ \mu\partial_\mu,
    \cal O]_\rho+
 [\gamma^ \mu\omega_\mu, \cal O]_\rho.
\end{equation}
On the one side, the Leibniz rule for the differential
operator $\partial_\mu$ together with \eqref{eq:92}  yields 
\begin{equation*}
  [\gamma^ \mu\partial_\mu, \cal O]_\rho\psi=
  \gamma^\mu\partial_\mu\cal O\psi - \rho(\cal O) \gamma^
  \mu\partial_\mu\psi =  \gamma^\mu(\partial_\mu\cal O)\psi  +
  \gamma^\mu\cal O\partial_\mu\psi- \rho(\cal O)
  \gamma^\mu\partial_\mu\psi=\gamma^\mu(\partial_\mu\cal O)\psi.
\end{equation*}
On the other side, by \eqref{eq:92}, 
\begin{equation}
  [\gamma^ \mu\omega_\mu, \cal O]_\rho=\gamma^ \mu\omega_\mu\cal O-
  \rho(\cal O)\gamma^ \mu\omega_\mu=\gamma^\mu[\omega^\mu, \cal O] 
\end{equation}
vanishes by hypothesis. Hence the result. 
\end{proof}

This lemma applies in particular to the components $Q$ and $M$ of the
representation of the algebra~$\A$ in  \eqref{eq:27}. The slight
difference is that these components do not act on $L^2(\M, S)$, but on
$L^2(\M, S)\otimes \C^{32}$. With a slight abuse of
notation, we write 
\begin{equation}
\gamma^\mu Q\coloneqq  (\gamma^\mu \otimes \I_{16}) \, Q\; , \;  \partial_\mu Q\coloneqq  (\partial_\mu \otimes \I_{16}) Q
\label{eq:117}
\end{equation}
and similarly for $M$.
\begin{corollary}
 \label{cor:twistcommut}
One has
  \begin{align}
    \label{eq:98}
&\gamma^\mu Q= \rho(Q)\gamma^\mu, \quad [\ds, Q]_\rho= -i\gamma^\mu\partial_\mu Q,\\
&\gamma^\mu M= \rho(M)\gamma^\mu, \quad [\ds, M]_\rho=-i\gamma^\mu\partial_\mu Q.
\end{align}
\end{corollary}
\begin{proof}
  From \eqref{eq:defQM} and omitting the internal indices (on which
  the action of $\gamma^\mu\otimes \I_{16}$ is trivial), one checks
  from the explicit form \eqref{EDirac} of the euclidean Dirac
  matrices that
  \begin{align}
    \label{eq:94}
    \gamma^\mu_E Q - \rho(Q)\gamma_E^\mu=
    \begin{pmatrix}
      0&\sigma^\mu\\ \tilde\sigma^\mu &0
    \end{pmatrix}_s^t \begin{pmatrix} Q_r& 0\\ 0& Q_l
    \end{pmatrix}_s^t-\begin{pmatrix} Q_l& 0\\ 0& Q_r
    \end{pmatrix}_s^t \begin{pmatrix} 0&\sigma^\mu\\ \tilde\sigma^\mu
      &0
    \end{pmatrix}_s^t =0.
  \end{align}
  The same holds true for the curved Dirac matrices \eqref{eq:111}, by
  linear combination. 

The commutation with the spin connection follows remembering
  that the latter is
  \begin{equation}
    \omega_\mu =\Gamma_\mu^{\rho\nu} \gamma_\rho\gamma_\nu=\Gamma_\mu^{\rho\nu}
    \begin{pmatrix}
      \sigma_\mu\tilde\sigma_\nu& 0\\ 0& \tilde\sigma_\mu\sigma_\nu
    \end{pmatrix}_s^t
    \label{eq:97}
  \end{equation}
  and so commutes with $Q$, which is diagonal in the
  $s,t$ indices an trivial in the $\dot s,\dot t$ indices.
\end{proof}

\subsection{Free $1$-form}
\label{sec:free1form}

  With the previous results, it is not difficult to compute a free $1$-form  \eqref{eq:40}.
\begin{lemma}
\label{lem:free-1-form}
A free $1$-form is
     \begin{equation}
      \label{eq:95}
      \slashed{A}
      =-i\gamma^\mu A_\mu \quad\text{ with }\quad A_\mu= \left(\begin{matrix}
          Q_\mu&0\\0&M_\mu
        \end{matrix}\right)_C^D,
    \end{equation}
where we use notations similar to \eqref{eq:117}, with
\begin{align}
  \label{eq:119}
  Q_\mu\coloneqq \sum_i \rho(Q_i)\partial_\mu R_i,\quad M_\mu=\sum_i \rho(M_i)\partial_\mu N_i
\end{align}
for $Q_i, M_i$ and $R_i, N_i$ the components of $a_i, b_i$ as in
(\ref{eq:27}, \ref{eq:defb}). 
  \end{lemma}
\begin{proof}
Omitting the summation index $i$, one has
  \begin{align}
    \slashed{A}=a\left[\slashed{D},b\right]_\rho&= 
\left(\begin{matrix}
        Q&0\\0&M
      \end{matrix}\right)_C^D
\left(\begin{matrix}
      [\ds, R]_\rho&0\\0&     [\ds, N]_\rho
      \end{matrix}\right)_C^D=\\
\nonumber    &=-i\left(\begin{matrix}
        Q&0\\0&M
      \end{matrix}\right)_C^D
\left(\begin{matrix}
        \gamma^\mu\partial_\mu R&0\\0&\gamma^\mu\partial_\mu N
      \end{matrix}\right)_C^D
    =-i\gamma^\mu\left(\begin{matrix}
                   \rho\left(Q\right)\partial_\mu
                   R&0\\0&\rho\left(M\right)\partial_\mu N
                 \end{matrix}\right)_C^D,
  \end{align}
where the last equalities follow from corollary \ref{cor:twistcommut}. Restoring the index $i$, one gets the result.
\end{proof}

By computing explicitly the components of $\slashed A$, one finds that
a free $1$-form is  parametrised by two complex
fields $c^r_\mu$, $c^l_\mu$, two 
quaternionic fields $q_\mu^r$, $q_\mu^l$ and two $M_3(\C)$-valued
fields $m^r_\mu$,~$m^l_\mu$.
\begin{proposition}
\label{prop:free1formself}
The components $Q_\mu, M_\mu$ of  $\slashed A$ in \eqref{eq:95}
are
\begin{equation}
  \label{eq:120}
  Q_\mu = \delta_{\dot s I}^{\dot t J}
  \begin{pmatrix}
    Q_\mu^r& \\ & Q_\mu^l
  \end{pmatrix}_s^t,\quad 
M_\mu =\delta_{\dot s}^{\dot t}
  \begin{pmatrix}
    M_\mu^r & \\  &M_\mu^l
  \end{pmatrix}_s^t
\end{equation}
where
\begin{align}
  \label{eq:90}
 & Q_\mu^r=
  \begin{pmatrix}
    {\sf c}^r_\mu &\\ & q^r_\mu
  \end{pmatrix}_\alpha^\beta,\quad   Q_\mu^l =
  \begin{pmatrix}
   {\sf c}_\mu^l &\\ & q^l_\mu 
  \end{pmatrix}_\alpha^\beta
\end{align}
for ${\sf c}_\mu^r =
\begin{pmatrix}
  c_\mu^r &\\ &\bar c_\mu^r
\end{pmatrix}$, ${\sf c}_\mu^l =
\begin{pmatrix}
  c_\mu^l &\\ &\bar c_\mu^l
\end{pmatrix}$
and
\begin{align}
\label{eq:90-2}&M_\mu^r=\begin{pmatrix}
    {\sf m}^r_\mu\otimes \I_2 & 0 \\
    0 &{\sf m}_\mu^l\otimes \I_2
  \end{pmatrix}_\alpha^\beta,\; 
M_\mu^{l}=\begin{pmatrix}
    {\sf m}_\mu^l \otimes \I_2 & 0 \\
    0 &{\sf m}_\mu^r\otimes \I_2
  \end{pmatrix}_\alpha^\beta
\end{align}
for
 $ {\sf m}_\mu^r=
  \begin{pmatrix}
    c_\mu^r&\\ & m_\mu^r
  \end{pmatrix}_I^J$, 
 ${\sf m}_\mu^l=
  \begin{pmatrix}
    c_\mu^l&\\ & m_\mu^l
  \end{pmatrix}_I^J$.

\noindent The complex, quaternionic and $M_3(\C)$-value fields $c_\mu^{r\slash l}, q_\mu^{r\slash l}, m_\mu^{r\slash l}$ are
defined in the proof.\end{proposition}
\begin{proof}

The form \eqref{eq:120}-\eqref{eq:90} of the components of $\slashed
A$ follows calculating explicitly \eqref{eq:119} using
\eqref{eq:defQM}-\eqref{eq:35} for $Q_i, M_i$ and \eqref{eq:31} for
$R_i, N_i$.
Omitting the $i$ index, one finds
\begin{equation}
  \label{eq:125}
  Q_\mu^r = Q_l \partial_\mu R_r, \quad Q_\mu^l= Q_r\partial_\mu
  R_l,\quad  M_\mu^r = M_l \partial_\mu N_r, \quad M_\mu^l= M_r\partial_\mu
  N_l.
\end{equation}
The first two equations yield \eqref{eq:90} with
\begin{align}
  \label{eq:121}
  & c_\mu^r=  c' \partial_\mu d,\quad  c_\mu^l=  c
  \partial_\mu d',\quad   q_\mu^r=  q \partial_\mu p',\quad q_\mu^l=  q'
  \partial_\mu p,
\end{align}
the last two ones yields ${\sf m}_\mu^r= \sf{m}'\partial_\mu
\sf{n},\quad \sf{m}_\mu^l= \sf{m}\partial_\mu \sf{n}'$, from which
\eqref{eq:90-2} follows with 
\begin{align*}
m_\mu^r= m'\partial_\mu n,\quad m_\mu^l= m\partial_\mu n'.
\end{align*}

\vspace{-.5truecm}
\end{proof}

\begin{corollary}
  A free  $1$-form $\slashed A$ is selfadjoint if and only if
\begin{equation}
  \label{eq:124}
  c_\mu^l=-\bar c_\mu^r,\quad  q_\mu^l=-(q_\mu^r)^\dag,  \quad
  m_\mu^l=-(m_\mu^r)^\dag.
\end{equation}
\end{corollary}
\begin{proof}
  From lemma \ref{lem:free-1-form} and corollary
  \ref{cor:twistcommut}, using that $\rho$ is a
  $*$-automorphism{\footnote{In a twisted spectral triple the
      automorphism is not necessarily involutive. What is asked is the
      regularity condition 
    $\rho(a^*)=(\rho^{-1}(a))^*$. In our case since
    $\rho^{-1}=\rho$, the latter is equivalent to
    $\rho$ being a $*$-autormorphism.}}, 
  one has
\begin{equation}
\slashed{A}^\dagger=i(A^\mu)^\dagger \gamma^\mu=i \gamma^\mu \rho(A^\mu)^\dagger, \label{eq:118}
\end{equation}
so $\slashed A$ is selfadjoint if and only if $\gamma^\mu(\rho(A_\mu)^\dag
+A_\mu) =0$. Since $A_\mu$ is diagonal the $s, t$ indices, the sum
$\Delta_\mu\coloneqq \rho(A_\mu)^\dag+A_\mu$ is also diagonal with components
$\Delta_\mu^{r/l}$. Thus
\begin{equation}
  \label{eq:127}
  \gamma^\mu\Delta_\mu=
  \begin{pmatrix}
     0& \sigma^\mu\Delta_\mu^l \\ \tilde\sigma^\mu \Delta_\mu^r
  \end{pmatrix}_s^t.
\end{equation}
If this is zero, then for any $\gamma^\nu$
\begin{equation}
  \label{eq:128}
\gamma^\nu\gamma^\mu \Delta_\mu = 
  \begin{pmatrix}
 \sigma^\nu\tilde\sigma^\mu\Delta_\mu^r &0\\ 0&\tilde\sigma^\nu\sigma^\mu\Delta_\mu^l  
  \end{pmatrix}=0.
\end{equation}
Being $A_\mu$ -- hence $\Delta_\mu$ --  trivial in $\dot s,\dot t$,  and
since $\text{Tr}\, \tilde\sigma^\mu\sigma^\nu= 2\delta_{\mu\nu}$, the
partial trace on the $\dot s, \dot t$ indices of the expression above
yields $\Delta_\mu^r =\Delta_\mu^l=0$. Therefore  $\gamma^\mu(\rho(A_\mu)^\dag
+A_\mu) =0$ implies 
\begin{equation}
  \label{eq:129}
  \rho(A_\mu)^\dag=- A_\mu.
\end{equation}
The converse is obviously true. Consequently, $\slashed A_\mu$ is
selfadjoint if and only if \eqref{eq:129} holds true.

From \eqref{eq:95}, this is equivalent to $\rho(Q_\mu)^\dag = -Q_\mu$
and $\rho(M_\mu)^\dag = -M_\mu$ that is, from \eqref{eq:120},
\begin{equation}
\label{eq:selfadjoint}
(Q_\mu^l)^\dag=-Q_\mu^r \quad  \text{ and } (M_\mu^l)^\dag=-M_\mu^r.
\end{equation}
This is
equivalent to \eqref{eq:124}.
\end{proof}

\subsection{Identification of the physical degrees of freedom}
\label{subsec:physdegrees}

To identify the physical fields, one follows the non twisted
case \cite{Chamseddine:2007oz} and separates the real from the
imaginary parts. We thus define two real fields $a_\mu=
\Re c^r_\mu$ and 
$B_\mu=-\frac 2{g_1}\Im c^r_\mu$ ($g_1$ is a real constant and the signs are such to match the
notations of \cite{Connes:2008kx}, see remark \ref{remark:fitSM}), so that
\begin{equation}
\label{eq:defBmu}
c^r_\mu =a_\mu- i\frac{g_1}{2}B_\mu,\quad c_\mu^l= -\bar c^r_\mu = -a_\mu-i\frac{g_1}{2}B_\mu.
\end{equation}
Moreover, we denote $w_\mu$ and $-\frac{g_2}2 W^k$ for $k=1,2,3$ the real components of the quaternionic field $q^r_\mu$ on the basis
$\left\{\I_2,i\sigma_j\right\}$ of the (real) algebra of quaternions
(with $g_2$ another real constant), so that
\begin{equation}
\label{eq:defqmu}
q^r_\mu= w_\mu\,\I_2-i\frac{g_2}2W^k_\mu \sigma_k,\quad q_\mu^l =
  -(q_\mu^r)^\dag = -w_\mu \,\I_2  - i\frac{g_2}2 W^k_\mu \sigma_k.
\end{equation}
Finally, we write $m_\mu^r$ as the sum of a selfadjoint part $g_\mu =
\frac 12 (m_\mu^r + {m_\mu^r}^\dag)$ and an antiselfadjoint part $\frac 12
(m_\mu^r - {m_\mu^r}^\dag)$. We denote $V_\mu^0,
\frac{g_3}2 V_\mu^m$  the real-field components of the latter on the basis 
$\left\{i\I_3,i\lambda_m\right\}$ of the (real) vector space of
antiselfadjoint $3\times 3$ complex matrices (with $\left\{\lambda_m, m=1\ldots
  8\right\}$ the Gell-Mann matrices and $g_3$ a real constant), so that
\begin{align}
\label{eq:defmmur}
&m_\mu^r= g_\mu +iV^0_\mu \,\I_3+ i\frac{g_3}2 {V_\mu^m} \lambda_m,\\
\label{eq:defmmul}
&m_\mu^l=-(m_\mu^r)^\dag =-g_\mu + i{V^0_\mu} \,\I_3+i\frac{g_3}2 {V_\mu^m} \lambda_m.
\end{align}

The cancellation of anomalies is imposed requiring the the unimodularity condition
\begin{equation}
  \label{eq:130}
  \Tr A_\mu=0.
\end{equation}
This yields the same condition as in the non-twisted case.
\begin{proposition}
  The unimodularity condition for a selfadjoint free $1$-form yields
\begin{equation}
  \label{eq:130-2}
V^0_\mu=\frac{g_1}{6}B_\mu.
\end{equation}
\end{proposition}
\begin{proof}
From proposition \ref{prop:free1formself}  one gets $\Tr A_\mu = \Tr Q_\mu + \Tr M_\mu$. On the one
  side (neglecting the $\dot s$ and $I$ indices)
  \begin{align}
    \label{eq:131}
     \Tr Q_\mu =     \Tr Q_\mu^r +      \Tr Q_\mu^l = c_\mu^r + \bar
    c_\mu^r +\Tr q_\mu^r +   c_\mu^l + \bar
    c_\mu^l+\Tr q_\mu^l   
  \end{align}
vanishes by \eqref{eq:124}, when one notices that $\Tr q^\dag =\Tr q$
for any quaternion $q$. On the other side
\begin{align}
  \Tr M_\mu =  \Tr M_\mu^r + \Tr M_\mu^l &= 4\Tr {\sf m}_\mu ^r +  4\Tr
                                           {\sf m}_\mu ^l=\nonumber\\
  \label{eq:132}
&=4(c_\mu^r +\Tr m_\mu^r + c_\mu^l + \Tr m_\mu^l)=4(-ig_1B_\mu + 6iV^0_\mu)
\end{align}
where we use $c_\mu^r+c_\mu^l=-ig_1B_\mu$ and $m_\mu^r+
m_\mu^l=2iV_\mu^0\I_3 + 2ig_3V_\mu^m\lambda_m$, remembering then that
the Gell-Mann matrices are traceless. Hence \eqref{eq:130} is equivalent
to \eqref{eq:130-2}.
\end{proof}

Let us summarise the results of this section in the following
\begin{proposition}
\label{prop:freeunimod}
A  unimodular selfadjoint free $1$-form $\slashed A$ is parametrised
by
\begin{itemize}
\item two real $1$-form fields $a_\mu$, $w_\mu$ and  a selfadjoint
  $M_3(\C)$-value field $g_\mu$,
\item  a $\frak u(1)$-value field $iB_\mu$,
  a $\frak{su}(2)$-value field $iW_\mu$ and a $\frak{su}(3)$-value
  field $iV_\mu$.
\end{itemize}
  \end{proposition}
\begin{proof}
Collecting the previous results, denoting $W_\mu\coloneqq W_\mu^k\sigma_k$ and $V_\mu\coloneqq 
V_\mu^m\lambda_m$,  one has
\begin{align}
\label{summary1}
c_\mu^r&=a_\mu-i\frac{g_1}{2}B_\mu, &   &c_\mu^l=-a_\mu-i\frac{g_1}{2}B_\mu,\\
\label{summary2}
q_\mu^r&=w_\mu\, \I_2-i\frac{g_2}{2}\, W_\mu ,& &q_\mu^l=-w_\mu \, \I_2-i\frac{g_2}{2}\, W_\mu,\\
m_\mu^r&=g_\mu+i\left(\frac{g_1}{6}B_\mu \,\I_3+\frac{g_3}{2}\, V_\mu\right),& & m_\mu^l=-g_\mu+i\left(\frac{g_1}{6}B_\mu \I_3+\frac{g_3}{2}V_\mu\right).
\label{summary3} 
 \end{align}

On the one side,  $a_\mu, w_\mu$ are in $C^\infty(\M, \mathbb R)$ and $g_\mu={g_\mu}^\dag$ is in  $C^\infty(\M, M_3(\C))$. 
On the other side, since $B_\mu$ is
real, $iB_\mu\in C^\infty(\M, i\mathbb R)$ is a $\frak
u(1)$-value~field. The Pauli matrices span the space of traceless
$2\times 2$ selfadjoint matrices, thus the field $iW_\mu$
takes value in the set of antiselfadjoint such matrices, that is
$\frak{su}(2)$. Finally, the real span of the Gell-Mann matrices is
the space of traceless selfadjoint elements of $M_3(\C)$, hence $iV_\mu$ is a $\frak{su}(3)$-value field.
 \end{proof}

  In the non-twisted case, the primed and unprimed quantities in
  \eqref{eq:121} and the next equation are equal, meaning that the right and left components of the fields
 \eqref{summary1}-\eqref{summary3} are equal, hence
  \begin{equation}
    \label{eq:133}
    a_\mu= w_\mu=g_\mu=0.
 \end{equation}That the twisting produces some extra $1$-form fields has already been pointed out
for manifolds in \cite{Lett.}, and for electrodynamic in
\cite{Martinetti:2019aa}. Actually, such a field (improperly called
vector field) appeared initially in the twisted version of the Standard
Model presented in \cite{buckley}, but its precise structure --
a collection of three selfadjoint fields $a_\mu, w_\mu, g_\mu$,  each associated with a gauge field of
the Standard Model -- had not
been worked out there. 

In the minimal twist of electrodynamics, there is only
one such field (associated with the $U(1)$ gauge symmetry). By studying
the fermionic action, it gets interpreted as
energy-momentum $4$-vector in lorentzian signature. Whether such an
interpretation still holds for $a_\mu, w_\mu, g_\mu$ will be
investigated in a forthcoming paper \cite{Manuel-Filaci:2020aa}.
\begin{remark}
\label{remark:fitSM}
In the non-twisted case, the fields $B_\mu, W_\mu$ and $V_\mu$ coincide with those of the spectral triple of the
Standard Model. More precisely, within the conditions of
\eqref{eq:133}, then 
\begin{itemize}
\item our $c_\mu^r = c_\mu^l$ coincides with $-i\Lambda_\mu$ of
  \cite[\S 15.4]{Connes:2008kx}{\footnote{Beware that $\ds_M$ in the
    formula of $\Lambda$ is $i\gamma^\mu\partial_\mu$ \cite[1.580]{Connes:2008kx}, so that
    $\Lambda=\Lambda_\mu\gamma^\mu$ is the $U(1)$ part of $-\slashed
    A$, meaning that $\Lambda_\mu$ is the $U(1)$ part of $iA_\mu$.}} The selfadjointness condition
  \eqref{eq:124} then implies that $\Lambda_\mu$ is real, in agreement
  with \cite{Connes:2008kx}. Then $B_\mu =\frac 2{g_1}\Lambda_\mu$ as
  defined in \cite[1.729]{Connes:2008kx} coincides with our
$B_\mu=-i\frac 2{g_1}c_\mu^r=-i\frac 2{g_1}c_\mu^l$ as defined in
\eqref{eq:defBmu}.

\item our $q_\mu^r = q_\mu^l$ coincides with $-iQ_\mu$ of
  \cite[\S 15.4]{Connes:2008kx}. The selfadjointness condition
  \eqref{eq:124} then implies that $Q_\mu$ is selfadjoint,  in agreement
  with \cite{Connes:2008kx}. Then $W_\mu =\frac 2{g_2}Q_\mu$ as
  defined in \cite[1.739]{Connes:2008kx} coincides with our
$W_\mu= W_\mu^k\sigma_k=i\frac 2{g_2}q_\mu^r=i\frac 2{g_2}q_\mu^l$ in
\eqref{eq:defqmu}.

\item the identification of our $V_\mu$ with the one of the
  non-twisted case is made after proposition \ref{prop:freetwisfluc}.
\end{itemize}

\end{remark}

\begin{remark}
\label{rem:ident-phys-degr}
If one does not impose the selfadjointness of $\slashed A$, then one
  obtains two copies of the bosonic contents of the Standard Model,
  acting independently on the right and left components of Dirac
  spinors. Whether this may yield physically meaningful models should
  be investigated elsewhere (considering to remove also
  the selfadjointness of the finite part of the fluctuation).
\end{remark}

\subsection{Twisted  fluctuation of the free Dirac operator}
\label{subsec:twistfreefluc}

We now compute the free part \eqref{eq:115} of the twisted
fluctuation. 

\begin{proposition}
\label{prop:freetwisfluc}
 A twisted fluctuation of the free Dirac operator $\slashed D$ is $D_Z=\slashed D + Z$
where 
\begin{equation}
\label{eq:defZ}
Z=\slashed{A}+ J\slashed{A} J^{-1}=-i\gamma^\mu
\begin{pmatrix}
Z^\mu&0\\
0&\overline{Z^\mu}
\end{pmatrix}_C^D
\text{ with }
\quad   Z_\mu =\gamma^5\otimes X_\mu + \I_4\otimes iY_\mu,
\end{equation}
in which $X_\mu$ and $Y_\mu$ are selfadjoint
$\A_{\text{SM}}$-value tensor fields on $\M$ with components
\begin{equation}
(X_\mu)_{\dot 1 I}^{\dot 2 J}= (X_\mu)_{\dot 2 I}^{\dot 1
  J}=(Y_\mu)_{\dot 1 I}^{\dot 2 J}= (Y_\mu)_{\dot 2 I}^{\dot 1J}=0,
\end{equation}
and 
\begin{align}
  \label{eq:120ZXY}
 &(X_\mu)_{\dot 1I}^{\dot 1J}= (X_\mu)_{\dot 2I}^{\dot 2J}=
\begin{pmatrix}
2a_\mu & \\& a_\mu\I_3 + g_\mu\end{pmatrix} _I^J,\\
 &(Y_\mu)_{\dot 1I}^{\dot 1J} = 
  \begin{pmatrix}
    0& \\ &-\frac{2g_1 }3B_\mu\I_3 - \frac{g_3}2 V_\mu
  \end{pmatrix},\quad 
 (Y_\mu)_{\dot 2I}^{\dot 2J} = 
  \begin{pmatrix}
    g_1 B_\mu& \\ &\frac{g_1}3 B_\mu\I_3 - \frac{g_3}2 V_\mu
  \end{pmatrix}
\\
& (X_\mu)_{aI}^{bJ}= 
\begin{pmatrix}
\delta_a^b \left(w_\mu -a_\mu\right) &  \\&\delta_a^b w_\mu\I_3 -
g_\mu  \end{pmatrix} _I^J,\\
& (Y_\mu)_{aI}^{bJ}= 
\begin{pmatrix}
\delta_a^b\frac{g_1}2B_\mu -\frac{g_2}2(W_\mu)_a^b& \\
&-\delta_a^b\left(\frac{g_1}6B_\mu\I_3 +\frac{g_3}2V_\mu \right)-\frac{g_2}2(W_\mu)_a^b \I_3\end{pmatrix} _I^J.
\end{align}
\end{proposition}
\begin{proof}
With $J=-J^{-1}$ as defined in \eqref{eq:108} one has
\begin{align*}
J\slashed{A}J^{-1}&=-J(-i\gamma^\mu A_\mu) J^{-1} =-iJ\gamma^\mu A_\mu
                    J^{-1} =i\gamma^\mu J A_\mu
                    J^{-1} =
&=i\gamma^\mu\left(\begin{matrix}
{\cal J} M_\mu \J^{-1} & 0\\ 0& {\cal J} Q_\mu \J^{-1}
\end{matrix}\right)_C^D
\end{align*}
where we use that $J$ is anti-linear and anticommutes with
$\gamma^\mu$ (lemma \ref{lemma:realtwist}). Noticing
that  ${\cal J} M_\mu \J^{-1}= -\bar M_\mu$ and  ${\cal J} Q_\mu
\J^{-1}= -\bar Q_\mu$ (this is shown as in \eqref{eq:JMJ},
\eqref{eq:JQJ}), one obtains
\begin{equation}
  Z_\mu = Q_\mu + \bar M_\mu.
\end{equation}
Explicitly, $Z_\mu=
\begin{pmatrix}
  Z_\mu^r& \\ & Z_\mu^l
\end{pmatrix}$ where,  using the explicit forms \eqref{eq:90}
and \eqref{eq:90-2} of $Q_\mu^r$ and $M_\mu^r$,
\begin{align}
\label{eq:38}
  Z_\mu^r = \delta_{\dot s I}^{\dot t J}Q_\mu^r + \delta_{\dot
  s}^{\dot t}\,\overline{M_\mu^r}=
 \delta_{\dot s}^{\dot t} \begin{pmatrix}
  {\sf c}_\mu^r  \delta_I^J +\delta_{\dot a}^{\dot b}\, \overline{{\sf m}_\mu^r}\\
&    q_\mu^r\delta_I^J +\delta_a^b\overline{{\sf m}_\mu^l}
  \end{pmatrix}_\alpha^\beta
\end{align}
and 
\begin{equation}
\label{eq:30}
   Z_\mu^l = \delta_{\dot s I}^{\dot t J}Q_\mu^l + \delta_{\dot
  s}^{\dot t}\,\overline{M_\mu^l}.
\end{equation}

The
components of the matrix in the r.h.s. of \eqref{eq:38}  are
\begin{equation}
(Z_\mu^r)_{\dot a I}^{\dot b J}={\sf c}_\mu^r\delta_I^J+\delta_{\dot a}^{\dot b} \overline{{\sf m}_\mu^r}=
    \begin{pmatrix}
c_\mu^r \delta_I^J+ \overline{{\sf m}_\mu^r} & \\
&\overline{c_\mu^r} \delta_I^J+ \overline{{\sf m}_\mu^r} 
    \end{pmatrix} _{\dot a}^{\dot b}\end{equation}
with $(Z_\mu^r)_{\dot 1 I}^{\dot 2 J}= (Z_\mu^r)_{\dot 2 I}^{\dot 1 J}=0$
and, using Proposition \ref{prop:freeunimod}, 
\begin{align}
\nonumber
(Z_\mu^r)_{\dot 1 I}^{\dot 1J} =c_\mu^r \delta_I^J+ \overline{{\sf m}_\mu^r} &= \begin{pmatrix}
2a_\mu & \\
& (a_\mu - i\frac{g_1 }2B_\mu)\I_3 + g_\mu -i\left(\frac{g_1}6B_\mu\I_3 +
\frac{g_3}2 V_\mu\right)  \end{pmatrix} _I^J=: (X_\mu^r)_{\dot 1
  I}^{\dot 1 J} + i  (Y_\mu^r)_{\dot 1 I}^{\dot 1 J}\\
(Z_\mu^r)_{\dot 2I}^{\dot 2 J} =\overline{c_\mu^r} \delta_I^J+ \overline{{\sf m}_\mu^r} &= \begin{pmatrix}
2a_\mu + ig_1B_\mu& \\
\nonumber
& (a_\mu + i\frac{g_1 }2B_\mu)\I_3 + g_\mu -i\left(\frac{g_1}6B_\mu\I_3 +
\frac{g_3}2 
V_\mu\right)  \end{pmatrix} _I^J =:(X_\mu^r)_{\dot 2 I}^{\dot 2 J} + i
  (Y_\mu^r)_{\dot 2 I}^{\dot 2 J};
\end{align}
and 
\begin{equation}
(Z_\mu^r)_{a I}^{ b J}=q_\mu^r\,\delta_I^J+\delta_{a}^{b} \overline{{\sf m}_\mu^l}=
    \begin{pmatrix}
(q_\mu^r)_1^1\, \delta_I^J+ \overline{{\sf m}_\mu^l} & (q_\mu^r)_1^2\,\delta_I^J\\[5pt]
(q_\mu^r)_2^1\,\delta_I^J&\overline{q_\mu^r}_2^2\,
 \delta_I^J+ \overline{{\sf m}_\mu^l} 
    \end{pmatrix} _{a}^{b}\end{equation}
with
\begin{align*}
\nonumber
(Z_\mu^r)_{aI}^{aJ} &=(q_\mu^r)_a^a\, \delta_I^J+ \overline{{\sf m}_\mu^l}=\nonumber\\&= \begin{pmatrix}
w_\mu -i\frac{g_2}2 
(W_\mu)_a^a -a_\mu + i\frac{g_1B_\mu}2& \\
\nonumber&\left(w_\mu -i\frac{g_2}2 
(W_\mu)_a^a\right)\I_3 - g_\mu - i\left(\frac{g_1B_\mu}6\I_3 +\frac{g_3}2
V_\mu\right)\end{pmatrix} _I^J,\\&=: (X_\mu^r)_{aI}^{aJ} + i  (Y_\mu^r)_{aI}^{aJ},\\
\nonumber
(Z_\mu^r)_{a}^{b\neq a} &=(q_\mu^r)_a^b\, \delta_I^J= \begin{pmatrix}
 -i\frac{g_2}2 
(W_\mu)_a^b& & \\
& -i\frac{g_2}2 
(W_\mu)_a^b\I_3 \end{pmatrix} _I^J = (X_\mu^r)_{aI}^{bJ} + i
  (Y_\mu^r)_{aI}^{bJ}.\end{align*}

The matrices $X_\mu^r$ and $Y_\mu^r$ defined by the equations above
are selfadjoint (notice that $W_\mu$ as defined in Prop. \ref{prop:freeunimod} is
selfadjoint) and such that
\begin{equation}
Z_\mu^r = X_\mu^r + iY_\mu^r.
\end{equation}
 The
selfadjointness condition \ref{eq:selfadjoint} applied to
\eqref{eq:30} yields
\begin{equation}
  Z_\mu^l = - (Z_\mu^r)^\dag= -X_\mu^r
  + i Y_\mu^r.
\end{equation} 
In other terms, $Z_\mu^l = X_\mu^l + i Y_\mu^l$ with 
\begin{equation}
  X_\mu^l= -X_\mu^r, \quad   Y_\mu^l= Y_\mu^r.
\end{equation}
Redefining $X_\mu\coloneqq X_\mu^r = -X_\mu^l$, $Y_\mu\coloneqq Y_\mu^r =Y_\mu^l$,
one obtains the result.\end{proof}

We collect the components of $Z$ in appendix \ref{app:gauge}. There,
we also make explicit that $iY_\mu$ coincides exactly with the gauge
fields of the Standard Model (including the ${\frak su}(3)$ gauge field $V_\mu$). Thus the twist does not modify the gauge
content of the model. What it does is to add the selfadjoint part
$X_\mu$ whose action on spinors breaks chirality. As shown in the next
section, this field is invariant under a gauge transformation. 

\section{Gauge Transformations}
\label{sec:gaugetransform}
\setcounter{equation}{0} 

A gauge transformation is implemented
  by an action of the group ${\cal U}(\A)$ of unitary elements of
  $\A$, both on the Hilbert space and on the Dirac operator.
On a twisted spectral triple, these actions
have been worked out in \cite{Devastato:2018aa,Landi:2017aa} and
 consist in a twist of the original formula of
 Connes~\cite{Connes:1996fu}, later generalised without the first
 order condition in \cite{Chamseddine:2013fk}. Explicitly,  on the
  Hilbert space, the fermion fields transform under the adjoint action
  of ${\cal U}(\A)$ induced by the real structure, namely
\begin{equation}
\label{eq:gaugefermion}
\psi\rightarrow\Ad u \ \psi\coloneqq u\psi u=uu^\circ\psi=uJu^*
J^{-1}\psi, \quad u\in\cal U.
\end{equation}
On the other hand, the twisted-covariant Dirac
operator $D_{A}$ \eqref{eq:12bis} transforms under the twisted conjugate action
of $\Ad{u}$,
\begin{equation}
\label{eq:gaugedirac}
D_{A}\rightarrow\Ad\rho(u) \, D_{A}\Ad u^*.
\end{equation}
By  \cite[Prop.4.2]{Martinetti:2021aa}, the
  operator $D_A$, viewed as a function of the
  components $a_i, b_i$ of the twisted $1$-form  $A= A_{(1)}=
  \sum_ia_i[D, b_i]$, transforms under a gauge transformation in the operator
  $D_{A^u}$, where
  \begin{equation}
\label{eq:25}
A^u:= \rho(u) \left[D, u^* \right]_\rho +    \rho(u)
   A u^*.
 \end{equation} 
This is the twisted version of the law of transformation of generalised $1$-forms in ordinary spectral triples,
which in turn is a non-commutative generalisation of the law of
transformation of the gauge potential in ordinary gauge theories.\medskip

To write down the transformation $A\to A^u$, we need the explicit
form of  a unitary  $u$ of $\A$. The latter is a pair of  functions on
$\M$ with value in 
\begin{equation}
{\cal U}(\C)\times {\cal U}(\HHH)\times {\cal U}(M_3(\C))\simeq U(1)\times
SU(2)\times U(3).
\label{eq:24}
\end{equation}
Namely
\begin{equation}
\label{eq:24bis}
u=(e^{i\alpha},  e^{i\alpha'}, {\boldsymbol q}, {\boldsymbol q'}, \boldsymbol m, \boldsymbol m')
\end{equation}
with
\begin{equation}
  \label{eq:23}
  \alpha, \alpha'\in C^\infty(\M, \mathbb R), \quad {\boldsymbol q}, {\boldsymbol q}'\in
  C^\infty(\M, SU(2)),\quad {\boldsymbol m}, {\boldsymbol m}'\in C^\infty(\M, U(3)).  
\end{equation}
It acts on $\HH$ as
\begin{equation}
\label{eq:24ter}
u=\left(\begin{matrix}
\mathfrak {A}&\\&\mathfrak{B}
\end{matrix}\right)_C^D
\end{equation}
where, following  \eqref{eq:QMexplicit}-\eqref{eq:35}, one has $\frak
A_{\dot s s I\alpha}^{\dot t t J\beta}=  \delta_{\dot s I}^{\dot t J}
\frak A_{s\alpha}^{t\beta}$ and $\frak
B_{\dot s s I\alpha}^{\dot t t J\beta}=  \delta_{\dot s}^{\dot t}
\frak B_{s\alpha I}^{t\beta J}$ with 
\begin{equation}
\label{eq:defQMunit}
 \frak A_{s\alpha}^{t\beta}=\left(\begin{matrix}
(\frak A_r)_\alpha^\beta&\\ &(\frak A_l)_\alpha^\beta
\end{matrix}\right)_s^t ,
\qquad \frak B_{s\alpha I}^{t\beta J}=\left(\begin{matrix}
(\frak B_r)_{\alpha I}^{\beta J}&\\ & (\frak B_l)_{\alpha I}^{\beta J}
\end{matrix}\right)_s^t, 
\end {equation}
in which 
\begin{align}
\label{eq:defQunit}
\frak A_{r}=\left(\begin{matrix}
{\boldsymbol \alpha}&\\ & \boldsymbol q'
\end{matrix}\right)_\alpha^\beta ,\quad \frak A_l=\left(\begin{matrix}
{\boldsymbol \alpha'}&\\  &\boldsymbol q
\end{matrix}\right)_\alpha^\beta,
\end{align}
and
\begin{align}
\label{eq:defMunit}
\frak B_{r}=\begin{pmatrix}
    \sf m\otimes \I_2 & 0 \\
    0 &\sf m'\otimes \I_2
  \end{pmatrix}_\alpha^\beta,\; 
\frak B_{l}=\begin{pmatrix}
    \sf m' \otimes \I_2 & 0 \\
    0 &\sf m\otimes \I_2
  \end{pmatrix}_\alpha^\beta,
\end{align}
where we denote
\begin{equation}
  \label{eq:35unit}
  {\boldsymbol \alpha}\coloneqq 
  \begin{pmatrix}
    e^{i\alpha} & \\ &  e^{-i\alpha}
  \end{pmatrix},\; {\sf m}\coloneqq 
  \begin{pmatrix}
   e^{i\alpha} & \\ &\boldsymbol m
  \end{pmatrix}_I^J,\; {\boldsymbol \alpha'}\coloneqq 
  \begin{pmatrix}
     e^{i\alpha'} & \\ &  e^{-i\alpha'} 
  \end{pmatrix},\; {\sf m'}\coloneqq 
  \begin{pmatrix}
    e^{i\alpha'}& \\ &\boldsymbol m'
  \end{pmatrix}_I^J.
\end{equation}

\subsection{Gauge sector}
\label{subsec:gaugegauge}

A twisted gauge
transformation \eqref{eq:gaugedirac} does not necessarily preserve the selfadjointness of
the Dirac operator (because the action of the unitary is twisted on
the left, not on the right). Equivalently, $A^u$ in \eqref{eq:25} is not necessarily selfadjoint, 
even though one starts with a selfadjoint~$A$.

 This may seem as a weakness of the twisted case, since in
the non-twisted case selfadjointness is  preserved.
Actually the possibility to lose selfadjointness allows to implement
Lorentz symmetry and yields -- at
least for electrodynamics \cite{Martinetti:2019aa} -- an interesting interpretation of the
component $X_\mu$ of the free fluctuation $Z$ of proposition
\ref{prop:freetwisfluc} as a four-vector
energy-impulsion.

However, regarding the gauge part of the Standard Model which -- as
shown below -- is fully encoded in the component $iY_\mu$ of $Z$, it is
rather natural to ask the selfadjointness of the free $1$-form $\slashed A$ to be preserved. This reduces
the choice of unitaries to pair of elements of \eqref{eq:24} equal up to a constant.

\begin{proposition}
\label{prop:uselfadjoint}

A unitary $u$ whose action \eqref{eq:25} preserves the selfadjointness
of any unimodular selfadjoint free $1$-form $\slashed A$ is given by
\eqref{eq:24bis} with
\begin{equation}
\label{eq:uselfadjoint}
  \alpha' = \alpha + K,\quad {\boldsymbol q} = {\boldsymbol q}',\quad {\boldsymbol m}' = {\boldsymbol m}.
\end{equation}
The components \eqref{eq:95} of $\slashed A$ then transform as
\begin{align}
\label{gaugetransf1}
  &  c_\mu^r \longrightarrow   c_\mu^r -i\partial_\mu \alpha , &
  &c_\mu^l \longrightarrow  c_\mu^l -i\partial_\mu \alpha,\\
\label{gaugetransf2}
&q_\mu^r \longrightarrow  {\boldsymbol q}\, q_\mu^r\,  {\boldsymbol q}^\dag  + {\boldsymbol q}
  \left(\partial_\mu {\boldsymbol  q}^\dag\right),
& &q_\mu^l \longrightarrow  {\boldsymbol q}\, q_\mu^l\, {\boldsymbol q}^\dag  +
    {\boldsymbol q} \left(\partial_\mu {\boldsymbol q}^\dag\right),\\
\label{gaugetransf3}
&m_\mu^r \longrightarrow  {\boldsymbol m}\, m_\mu^r {\boldsymbol m}^\dag  + {\boldsymbol m}
  \left(\partial_\mu {\boldsymbol m}^\dag\right),& &m_\mu^l \longrightarrow
                                             {\boldsymbol m} \, m_\mu^l {\boldsymbol m}^\dag  + {\boldsymbol m}
                                             \left(\partial_\mu {\boldsymbol m}^\dag\right).
\end{align}
 \end{proposition}
\begin{proof}
From corollary \ref{cor:twistcommut} one
has (with the same abuse of notations \eqref{eq:117}, now with $\I_{32}$)
\begin{equation}
  \label{eq:6}
  \slashed A^u = \rho(u)\left( [\slashed D, u^*]_\rho + \slashed A u^* \right)=
  -i\gamma^\mu\left(u\left(\partial_\mu u^*\right) + u A_\mu u^*\right).
\end{equation}
Using the explicit forms \eqref{eq:24ter} of $u$ and \eqref{eq:95} of $A_\mu$, one finds
\begin{align}
  \slashed A^u = -i\gamma^\mu
  \begin{pmatrix}
  \frak A \left(\partial_\mu \frak A^\dag\right) + \frak A Q_\mu \frak
    A^\dag & 0 \\0&  \frak B \left(\partial_\mu \frak B^\dag \right) + \frak B M_\mu \
    \frak B^\dag
  \end{pmatrix}_C^D
\end{align}
meaning that a gauge transformation is equivalent to the
transformation
\begin{align}
  Q_\mu \longrightarrow \frak A \left(\partial_\mu \frak A^\dag \right)+ \frak A Q_\mu \frak
    A^\dag , \quad M_\mu \longrightarrow \frak B\left(\partial_\mu \frak B^\dag\right) + \frak B M_\mu \frak
    B^\dag.
\end{align}
From \eqref{eq:90} and \eqref{eq:90-2}, these equations are equivalent to 
\begin{align}
\label{pregaugetransf1}
&  c_\mu^r \longrightarrow e^{i\alpha} \partial_\mu e^{-i\alpha} + c_\mu^r =
  c_\mu^r -i\partial_\mu \alpha , & &c_\mu^l \longrightarrow  c_\mu^l
                                      -i\partial_\mu \alpha',\\
\label{pregaugetransf2}
&q_\mu^r \longrightarrow  {\boldsymbol q'}\, q_\mu^r\,  {\boldsymbol  q'}^\dag  + {\boldsymbol
  q'} \left(\partial_\mu {\boldsymbol q'}^\dag\right),
& &q_\mu^l \longrightarrow  {\boldsymbol q}\, q_\mu^l\, {\boldsymbol q}^\dag  +
    {\boldsymbol q} \left(\partial_\mu {\boldsymbol q}^\dag\right),\\
\label{pregaugetransf3}
&m_\mu^r \longrightarrow  {\boldsymbol m}\, m_\mu^r {\boldsymbol m}^\dag  + {\boldsymbol m}
  \left(\partial_\mu {\boldsymbol m}^\dag\right),& &m_\mu^l \longrightarrow
                                             {\boldsymbol m'} m_\mu^l {\boldsymbol m'}^\dag  + {\boldsymbol m'} \left(\partial_\mu {\boldsymbol m'}^\dag\right).
\end{align} 
For any unitary operator $\frak q$, one has that $\frak q \left(\partial_\mu {\frak q}^\dag\right)  =\frak
q [\partial_\mu, {\frak q}^\dag]$ is anti-hermitian (being
$\partial_\mu$ anti-hermitian as well). Hence, beginning with a
selfadjoint $\slashed A$ as in \eqref{eq:124},  requiring that $\slashed
A^u$ be selfadjoint is equivalent to
\begin{align}
  &\partial_\mu \alpha' =  \partial_\mu \alpha,\\
& {\boldsymbol q} \,  q_\mu^l {\boldsymbol q}^\dag + {\boldsymbol q}\left(\partial_\mu {\boldsymbol
  q}^\dag\right)= {\boldsymbol q}'  q_\mu^l {\boldsymbol q'}^\dag+
{\boldsymbol  q'}\left(\partial_\mu {\boldsymbol q'}^\dag\right),\\
&{\boldsymbol m'} \,  m_\mu^l {\boldsymbol m'}^\dag + {\boldsymbol m'}\left(\partial_\mu {\boldsymbol
  m'}^\dag\right)= {\boldsymbol m}\,  m_\mu^l {\sf m}^\dag + {\boldsymbol m}\left(\partial_\mu {\boldsymbol m}^\dag\right).
\end{align}
In particular, for $q_\mu^l$ the identity, the second of these equations
yields ${\boldsymbol q}\left(\partial_\mu {\boldsymbol q}^\dag\right)
= {\boldsymbol q'}\left(\partial_\mu {\boldsymbol q'}^\dag\right)$ for any $\boldsymbol q, \boldsymbol q'$. Hence for any  $q_\mu^l$ one has  ${\boldsymbol q} \,  q_\mu^l {\boldsymbol
  q}^\dag = {\boldsymbol q}'  q_\mu^l {\boldsymbol q'}^\dag$. This means that ${\boldsymbol q'}^\dag
\boldsymbol q$ is in the centre of $\HHH$. Being a unitary, ${\boldsymbol q'}^\dag
\boldsymbol q$ is thus the identity. So $\boldsymbol q =
  \boldsymbol q'$. Similarly,  one gets that ${\boldsymbol m'}^\dag
  \boldsymbol m$ is in
  the centre of $M_3(\C)$, that is a multiple of the identity.  Being
  unitary, ${\boldsymbol m'}^\dag \boldsymbol m$ can only be the identity, hence
  ${\boldsymbol m}' = {\boldsymbol m}$. Thus
  (\ref{pregaugetransf1}-\ref{pregaugetransf3}) yield the result.\end{proof}

These transformations of the components of the free
$1$-form induce the following  transformations of the physical fields defined in
\eqref{summary1}-\eqref{summary3}. 
\begin{proposition}
\label{prop:gaugetransform1}
  Under a twisted gauge transformation that preserve the
  selfadjointness of a unimodular free $1$-form, the physical fields $a_\mu$
  and $w_\mu$ are invariant, $g_\mu$ undergoes an algebraic (i.e. non-differential) transformation
  \begin{equation}\label{eq:transfgg}
   g_\mu\longrightarrow  {\boldsymbol n} g_\mu {\boldsymbol  n}^\dag 
  \end{equation}
and the gauge fields transform as in the Standard Model
\begin{align}
\label{eq:transBB}
&B_\mu\longrightarrow B_\mu + \frac{2}{g_1}\partial_\mu \alpha,\\
\label{eq:transWW}
&W_\mu\longrightarrow {\boldsymbol  q} W_\mu {\boldsymbol  q}^\dag  +
  \frac{2i}{g_2} {\boldsymbol  q}\,( \partial_\mu {\boldsymbol  q}^
\dag),\\
\label{eq:transVV}
&V_\mu\longrightarrow {\boldsymbol  n} V_\mu {\boldsymbol  n}^\dag  - \frac{2i}{g_3} {\boldsymbol  n}\,( \partial_\mu {\boldsymbol  n}^
\dag),
  \end{align}
where $\boldsymbol  n= (\det \boldsymbol  m)^{-\frac 13} \boldsymbol  m$ is the $SU(3)$ part of
$\boldsymbol  m$. 
\end{proposition}
\begin{proof}
Applying the gauge transformations
\eqref{gaugetransf1}-\eqref{gaugetransf3} to the
physical fields defined through \eqref{summary1}-\eqref{summary3}, one
obtains 
\begin{align}
 &\pm a_\mu - i\frac{g_1}2 B_\mu \rightarrow \pm a_\mu - i\left(\frac{g_1}2
  B_\mu + \partial_\mu\alpha\right),\\
\label{eq:transfq}
&\pm w_\mu\I_2 - i\frac{g_2}2W_\mu \rightarrow  \pm w_\mu\I_2 -
  i\left(\frac{g_2}2 {\boldsymbol  q} W_\mu {\boldsymbol  q}^\dag  + i {\boldsymbol 
  q}(\partial_\mu {\boldsymbol  q}^\dag)\right),\\
\label{eq:transfm}
&\pm g_\mu +i\left(\frac{g_1}6 B_\mu \I_3 + \frac{g_3}2 V_\mu\right)
  \rightarrow  \pm {\boldsymbol  m} g_\mu {\boldsymbol m}^\dag +
  i\left(\frac{g_1}6
  B_\mu \I_3 + \frac{g_3}2 {\boldsymbol  m} \, V_\mu {\boldsymbol  m}^\dag -i {\boldsymbol  m}
  (\partial_\mu {\boldsymbol  m}^\dag)\right),
\end{align}
where the anti-selfadjointness of ${\boldsymbol q}(\partial_\mu {\boldsymbol  q}^\dag)$ and ${\boldsymbol  m}(\partial_\mu {\boldsymbol  m}^\dag)$
guarantee that the r.h.s. of \eqref{eq:transfq} and \eqref{eq:transfm}
is split into a selfadjoint and anti-selfadjoint part.
The first two equations above yield \eqref{eq:transBB}-\eqref{eq:transWW}
Writing ${\boldsymbol  m} = e^{i\theta}{\boldsymbol  n}$ with $e^{i\theta} =(\det m)^{\frac 13}$ and
${\boldsymbol  n}\in{SU}(3)$, then the right hand side of \eqref{eq:transfm}
becomes
\begin{equation}
  \pm {\boldsymbol  n} g_\mu {\boldsymbol  n}^\dag + i\left(  (\frac{g_1}6 B_\mu- \partial_\mu\theta )\I_3 +
    \frac{g_3}2 {\boldsymbol  n}V_\mu{\boldsymbol  n}^\dag - i{\boldsymbol  n}\partial_\mu{\boldsymbol  n}^\dag\right).
\end{equation}
where we use ${\boldsymbol  m}\partial_\mu{\boldsymbol  m}^\dag = -i\partial_\mu\theta +
{\boldsymbol  n}\partial_\mu {\boldsymbol  n}^\dag$. 
Requiring the unimodularity
condition to be gauge invariant forces to identify $-\theta$ with
$\frac{\alpha}3$, thus reducing the gauge group $U(3)$ to $SU(3)$.
This yields \eqref{eq:transfgg} and \eqref{eq:transVV}.
\end{proof}

\begin{remark}
  If one does not impose that the twisted gauge transformation
  preserves selfadjointness, then the left and right components of
  spinors transform independently. As explained in remark
  \ref{rem:ident-phys-degr}, the viability of such models should be
  explore elsewhere.
\end{remark}

\subsection{Scalar sector}
\label{gauge:scalar}

We now study the gauge transformation \ref{eq:25} of the scalar part of the
twisted $A_Y + A_M$ of the twisted $1$-form computed in section
\ref{sec:scalarsector}, beginning with the Yukawa part $A_Y$ in
\eqref{eq:45}.

\begin{lemma} 
\label{lem:gaugeHiggs}
Let $u$ be a unitary of $\A$ as in \eqref{eq:24bis}. One has
\begin{equation}
\label{eq:fluccal}
A_Y^u=   \rho(u) \left[\gamma^5\otimes D_Y, u^\dag \right]_\rho +    \rho(u) A_Y u^\dag =
      \begin{pmatrix}
         A^u& \\  &0
      \end{pmatrix} _C^D
\end{equation}
where
\begin{equation}
  A^u = \delta_{\dot s  I}^{\dot t J}
  \begin{pmatrix}
    (A^u)_r & \\ & (A^u)_l 
  \end{pmatrix}_s^t\, ,
\end{equation} 
with
\begin{align}
\label{eq:transfar}
    (A^u)_r&=\begin{pmatrix}
   0 & \bar{\sf  k^I} \left({\boldsymbol \alpha}' \left(H_1 + \I\right){\boldsymbol
      q}^{'\dag}- \I\right)\\\left( {\boldsymbol q}\left(H_2 + \I\right)
    {\boldsymbol \alpha}^\dag - \I\right) {\sf  k^I} & 0 
    \end{pmatrix},\\
  \label{eq:transfal}
 (A^u)_l&=-\begin{pmatrix}
   0 & \bar{\sf  k^I} \left({\boldsymbol \alpha} \left(H'_1 + \I\right){\boldsymbol
      q}^{\dag}- \I\right)\\\left( {\boldsymbol q'}\left(H'_2 + \I\right)
    {\boldsymbol \alpha}^{'\dag} - \I\right) {\sf  k^I} & 0 
    \end{pmatrix}
\end{align}
where $H_{1,2}$ are the components of $A_Y$, and $\boldsymbol \alpha$,
$\boldsymbol \alpha'$, $\sf q$, $\sf q'$ those of $u$.
 \end{lemma}
\begin{proof}
From the formula \eqref{eq:49} of $A_Y$ and
\eqref{eq:24ter}-\eqref{eq:defQMunit} of $u$, one gets
\begin{equation}
      \rho(u) A_Y u^\dag =
      \begin{pmatrix}
       \rho(\frak A) A \frak A^\dag & \\  &0
      \end{pmatrix} _C^D
\quad \text{ where }\quad 
   \rho(\frak A) A \frak A^\dag =\delta_{\dot s  I}^{\dot t J}
   \begin{pmatrix}
      \frak A_l A_r \frak A_r ^\dag& \\ &          \frak A_r A_l \frak A_l ^\dag
   \end{pmatrix}_s^t,
\end{equation}
where, using \eqref{eq:67} and  \eqref{eq:defQunit},
\begin{align}
\label{eq:ularur+}
  \frak A_l A_r \frak A_r ^\dag &=
  \begin{pmatrix}
    \boldsymbol \alpha ' & \\ &\boldsymbol q
  \end{pmatrix}
  \begin{pmatrix}
   & \bar{\sf k^I} H_1\\  H_2{\sf k^I}
  \end{pmatrix}  \begin{pmatrix}
    \boldsymbol \alpha ^\dag & \\ &\boldsymbol  q^{'\dag}
  \end{pmatrix}=\begin{pmatrix}
  &  \bar{\sf k^I}\boldsymbol \alpha' \, H_1 \boldsymbol q^{'\dag} \\ 
 {\boldsymbol q} H_2 \boldsymbol \alpha^{\dag} {\sf
   k^I}&\end{pmatrix}_\alpha^\beta,\\
\frak A_r A_l \frak A_l ^\dag &=
  -\begin{pmatrix}
    \boldsymbol \alpha & \\ &\boldsymbol q'
  \end{pmatrix}
  \begin{pmatrix}
   & \bar{\sf k}^I H'_1\\  H'_2{\sf k^I}
  \end{pmatrix}  \begin{pmatrix}
    \boldsymbol \alpha ^{'\dag} & \\ &\boldsymbol q^{\dag}
  \end{pmatrix}=-\begin{pmatrix}
  &  \bar{\sf k^I}\boldsymbol \alpha \, H'_1 \boldsymbol q^{\dag} \\ 
 {\boldsymbol q}' H'_2 \boldsymbol \alpha^{'\dag} {\sf k^I}&\end{pmatrix}_\alpha^\beta,
\end{align}
where we used that ${\sf k^I}$, $\bar{{\sf k^I}}$ commute with $\boldsymbol
\alpha$, $\boldsymbol\alpha'$ and their conjugates.

The computation of the twisted commutator part in \eqref{eq:fluccal} is similar to that of $A_Y$ in proposition
\ref{prop:diag1form}, with $a_i= \rho(u)$ and $b_i= u^\dag$ for $u$ as
in \eqref{eq:24bis}, that is 
 \begin{equation}
    \rho(u) \left[\gamma^5\otimes D_Y, u^\dag \right]_\rho =
    \begin{pmatrix}
      \frak U & \\ &0
    \end{pmatrix}_C^D
  \end{equation}
where
\begin{equation}
\label{eq:U}
      \frak U=  \delta_{\dot s I}^{\dot t J}
      \begin{pmatrix}
            \frak U_r & \\ &     \frak U_l
      \end{pmatrix}_s^t\;\text{ with }\; 
  \frak U_r =
\begin{pmatrix}
    & \bar{\sf k^I} {\frak H}_1\\ {\frak H}_2\sf k^I
  \end{pmatrix}_\alpha^\beta,\quad   \frak U_l =-
  \begin{pmatrix}
    & \bar{\sf k^I} {\frak H}'_1\\ {\frak H}'_2\sf k^I
  \end{pmatrix}_\alpha^\beta,\;
\end{equation}
in which ${\frak H}_{i=1, 2}$ and ${\frak H}'_{1,2}=\rho(\frak
H_{1,2})$ are given by \eqref{eq:65bis}
with (remembering \eqref{eq:35unit})
\begin{align}
{\sf c}=\boldsymbol \alpha',\; {\sf c}'=\boldsymbol \alpha,\; q= {\boldsymbol
  q}',\; q'= {\boldsymbol q}\; \text{ and } \;
{\sf d}= \boldsymbol \alpha^\dag,\; {\sf d'}= \boldsymbol \alpha^{'\dag},\;   
p={\boldsymbol q}^\dag,\, p'={\boldsymbol q}^{'\dag};
\end{align}
that is 
\begin{equation}
\label{eq:defH}
  {\frak H}_1=\boldsymbol \alpha' (\boldsymbol q^{'\dag} - \boldsymbol
  \alpha^{'\dag}),\;
  {\frak H}_2= {\boldsymbol q}(\boldsymbol \alpha^\dag - \boldsymbol q^\dag) \; \text{ and } \;
 {\frak H}'_1=\boldsymbol \alpha (\boldsymbol q^\dag - \boldsymbol \alpha^\dag),\;
  {\frak H}'_2= \boldsymbol q'(\boldsymbol \alpha^{'\dag} - \boldsymbol q^{'\dag}).
\end{equation}

Thus one obtains \eqref{eq:fluccal} with 
\begin{equation}
 A^u=  \frak U +   \rho(\frak A) A \frak A^\dag=   \delta_{\dot s  I}^{\dot t J}\begin{pmatrix}
            \frak U_r +   \frak A_l A_r \frak A_r ^\dag& \\ &     \frak U_l+  \frak A_r A_l \frak A_l ^\dag
      \end{pmatrix}_s^t.
\end{equation}
From
\eqref{eq:ularur+} and  \eqref{eq:U} one obtains the explicit forms of  $(A^u)^r$ and   $(A^u)^l$
\begin{align}
  (A^u)_r\coloneqq \frak U_r +   \frak A_l A_r \frak A_r ^\dag=\begin{pmatrix}
    & \bar{\sf  k^I} (\frak H_1 +  {\boldsymbol \alpha}' H_1{\boldsymbol
      q}^{'\dag})\\ (\frak H_2 +{\boldsymbol q}
    H_2   {\boldsymbol \alpha}^\dag)  {\sf  k^I}& 0 
    \end{pmatrix},\\
  (A^u)_l\coloneqq \frak U_l +   \frak A_r A_l \frak A_l ^\dag=-\begin{pmatrix}
    & \bar{\sf  k^I} (\frak H'_1 +  {\boldsymbol \alpha} H'_1{\boldsymbol
      q}^{\dag})\\(\frak H'_2 +{\boldsymbol q}'
    H'_2   {\boldsymbol \alpha}^{'\dag})  {\sf  k^I} & 0 
    \end{pmatrix}.
\end{align}
The final result follow substituting $\frak H_{1,2}$ with their
explicit formulas \eqref{eq:defH}.
\end{proof}

A unitary $u$ that preserves the
selfadjointness of the unimodular free $1$-form (Prop. \ref{prop:uselfadjoint}) also preserves the
selfadjointness of $A_Y$ if, and only if, $K=0$. Indeed, in that case $u$
is twist-invariant (i.e. $\boldsymbol q' = \boldsymbol
q$ and $\boldsymbol\alpha' =\boldsymbol\alpha$)  and one easily checks that for a selfadjoint 
$A_Y$ 
(that is $H_1^\dag = H_2=H_r$
and ${H'_1}^\dag = H'_2=H_l$  by corollary \ref{cor:higgs}) then $A_Y^u$ is selfadjoint as well. If $K\neq0$, then $H_1$
and $H_2$ undergo different gauge transformations, forbidding
$A_Y^u$ to be selfadjoint.
 For this reason, from now on we take $K=0$. With this caveat, the gauge transformation of lemma
\ref{lem:gaugeHiggs} then
reads as a law of transformation of the complex components
\eqref{eq:Higgscomp} of the
quaternionic fields $H_r$ and $H_l$.

\begin{proposition}
\label{prop:gaugehiggs}

Let $A_Y$ be a selfadjoint diagonal $1$-form parametrised by two
quaternionic field $H_r, H_l$. Under a gauge
transformation induced by a twist-invariant
unitary $u=(\alpha, \alpha, {\boldsymbol q}, {\boldsymbol q}, m, m)$, the components $\phi^r_{1,2}$, $\phi^l_{1,2}$ of $H_r$,  $H_l$ 
transform as
\begin{equation}
\label{eq:transHiggs}
  \begin{pmatrix}
    \phi_1^r + 1 \\ \phi_2^r
  \end{pmatrix}\longrightarrow {\boldsymbol q} \begin{pmatrix}
    \phi_1^r + 1 \\ \phi_2^r
  \end{pmatrix}e^{-i\alpha},\quad \begin{pmatrix}
    \phi_1^l + 1 \\ \phi_2^l
  \end{pmatrix}\longrightarrow {\boldsymbol q} \begin{pmatrix}
    \phi_1^l + 1 \\ \phi_2^l
  \end{pmatrix}e^{-i\alpha}.
\end{equation}
\end{proposition}
\begin{proof}
$A_Y$ being selfadjoint means that \eqref{eq:75} holds. A twist-invariant unitary satisfies
\eqref{eq:uselfadjoint} with $K=0$.  Under these conditions,
comparing the formula \eqref{eq:67} of
$A_Y$ with its gauge transformed counterpart
\eqref{eq:transfar}-\eqref{eq:transfal}, one finds that the fields $H_r$ and $H_l$ undergo the same transformation
\begin{align}
  H_r \longrightarrow {\boldsymbol q}\left(H_r + \I\right) \boldsymbol \alpha^\dag - \I,\qquad
 H_l \longrightarrow {\boldsymbol q}\left(H_l + \I\right)\boldsymbol\alpha^\dag - \I.
\end{align}
Written in components \eqref{eq:Higgscomp},  with $q_{ij}$ the
components of $\bf q$, 
these equations reads
\begin{align}
\phi_1^r \longrightarrow  q_{11}\, (\phi^r_1 + 1)e^{-i\alpha} + q_{12}\,
  \phi^r_2\,e^{-i\alpha}-1 ,\quad
 \phi_2^r \longrightarrow q_{21}\,\alpha(\phi^r_1 + 1)e^{-i\alpha} + q_{22}\,
  \phi^r_2\,e^{-i\alpha},
\end{align}
and similarly for  $\phi_{1,2}^l$. In
matricial form, these equations are nothing but
\ref{eq:transHiggs}.\end{proof}

The transformations \eqref{eq:transHiggs} are similar to  those of the
Higgs doublet in the  Standard Model (see e.g. \cite[Prop. 11.5]{Walterlivre}).
In the twisted version of the Standard Model, we thus obtain two
Higgs fields, acting independently on the left and right components of
the Dirac spinors. However, as we already mentioned, the two have no individual physical meaning on their own, since they only appear in the fermionic action through the linear combination $h=\left(H_r+H_l\right)/2$. Therefore there is actually only one physical Higgs doublet in the twisted case as well.
\smallskip

To conclude, we check that the scalar field $\sigma$ is gauge
invariant. As explained below \eqref{eq:gaugedirac}, this 
invariance  is not affected by the non-linear term, and is encoded within the
transformation 
\begin{equation}
\label{eq:25bbis}
A_M\rightarrow A_M^u= \rho(u) \left[\gamma^5\otimes D_M, u^\dag \right]_\rho +    \rho(u)
   A_M u^\dag,
 \end{equation} 
\begin{proposition}
\label{prop:gaugesigma}
Under a gauge transformation induced by a
twist-invariant unitary $u$, the real fields $\sigma_r$, $\sigma_l$
parameterising a self-adjoint off-diagonal fluctuation (proposition~\ref{proposi:offdiagfluc})
are invariant.
\end{proposition}
\begin{proof}
The result amounts to showing that $A_M$ is invariant under
\eqref{eq:25bbis}. 
Since $u=\rho(u)$ by hypothesis, the twisted-commutator in \eqref{eq:25}
coincides with the usual one
$\left[\gamma^5\otimes D_M, u^\dag \right]$ which is zero by
\eqref{eq:needfortwist}. The explicit forms \eqref{eq:68} of
$A_M$ and \eqref{eq:24ter}
of $u$ yields
\begin{equation}
uA_M u^\dagger=
\begin{pmatrix}
   & \frak A C \frak B^\dag \\  \frak B D \frak A^\dag &
\end{pmatrix}.
\end{equation}
From \eqref{81bbis}, one checks that $\frak A C \frak B^\dag$ has
components (omitting the global factor $k_R\delta_{\dot s}^{\dot t}$
and $\delta_I^J$ per ${\mathfrak A}_{r/l}$)
\begin{align}
\label{eq:gaugesigma}
&\frak A_r C_r \frak
B^\dag_r= \sigma\, \frak A_r \, \Xi_{I\alpha}^{J\beta}\, \frak
B^\dag_r= \sigma \Xi_{I\alpha}^{J\beta},\\
\label{eq:gaugesigma2} &\frak A_l C_l \frak
B^\dag_l=-\sigma'\, \frak A_l  \Xi_{I\alpha}^{J\beta} \frak B^\dag_l=-\sigma' \Xi_{I\alpha}^{J\beta},
\end{align}
where we use the explicit forms \eqref{eq:defQunit}-\eqref{eq:35unit}
of $\frak A, \frak B$ to get $\frak A_r \, \Xi_{I\alpha}^{J\beta}\, \frak
B^\dag_r =e^{i\alpha}\, \Xi_{I\alpha}^{J\beta}\,
e^{-i\alpha}\!~=\!~\Xi_{I\alpha}^{J\beta}$, and similarly for~\eqref{eq:gaugesigma2}.
Hence $uA_M u^\dagger= A_M$, and the result. \end{proof}

\section*{Conclusion}

We have worked out the field content of a twisted version of the spectral triple of the
Standard Model. The physical meaning of these fields will be made
precise by the computation of the fermionic action in the second part
of this work  \cite{Filaci:2021aa}, as well as the possibility of gauge transformations
induced by non twist-invariant unitaries, and their relation
with lorentzian signature.

 As shown in
\cite{Manuel-Filaci:2020aa}, the twisted first-order condition
needs to be violated in order to generate the extra scalar field $\sigma$. This
forbids to apply the twist-by-grading of
\cite{Landi:2017aa}, since the latter always preserves this
condition. However, this violation  has no
real importance, being reabsorbed in the definition of the components
of $\sigma$. In this sense, the
model presented here is the one that minimally violates the twisted
first-order condition.



\appendix
\addcontentsline{toc}{section}{Appendices}
\section*{Appendix}
\setcounter{equation}{0} 

\setcounter{section}{1}
 \subsection*{A.1 Dirac matrices and real structure}
 \label{app:Dirac}
Let $\sigma_{j= 1,2,3}$ be the Pauli matrices:
\begin{equation}
\label{Pauli}
		\sigma_1 = \left(\begin{array}{cc} 0 & 1 \\ 1 & 0 \end{array}\right)\!,
\qquad	\sigma_2 = \left(\begin{array}{cc} 0 & -i \\ i & 0 \end{array}\right)\!,
\qquad	\sigma_3 = \left(\begin{array}{cc} 1 & 0 \\ 0 & -1 \end{array}\right)\!.
\end{equation}%
In four-dimensional euclidean space, the Dirac matrices (in chiral
representation) are
\begin{equation}
\label{EDirac}
	\gamma^\mu_E =
	\left( \begin{array}{cc}
		0 & \sigma^\mu \\ \tilde\sigma^\mu & 0
	\end{array} \right)\!, \qquad
	\gamma^5_E \coloneqq  \gamma^1_E\,\gamma^2_E\,\gamma^3_E\,\gamma^0_E =
	\left( \begin{array}{cc}
		\mathbb{I}_2 & 0 \\ 0 & -\mathbb{I}_2 
	\end{array} \right)\!,
\end{equation}
where, for $\mu = 0,j$, we define
\begin{equation}
\label{eq:defsigmamu}
	\sigma^\mu \coloneqq  \left\{ \mathbb{I}_2, -i\sigma_j \right\}\!, \qquad
	\tilde\sigma^\mu \coloneqq  \left\{ \mathbb{I}_2, i\sigma_j \right\}\!.
\end{equation}

On a (non-necessarily flat) riemannian spin manifold, the Dirac
matrices are linear combinations of the euclidean ones,
\begin{equation}
\gamma^\mu =
e_\mu^\alpha \gamma^\alpha_E
\label{eq:111}
\end{equation}
 where $\left\{e_\mu^\alpha\right\}$ are the vierbein, which are real
   fields on $\M$. These Dirac matrices are no longer
   constant on $\M$. This is a general result of spin geometry that the charge conjugation
commutes with the spin derivative (see e.g.  \cite[Prop. 4.18]{Walterlivre}). For sake of
completeness, we check it explicitly for a four dimensional riemannian manifold:
\begin{lemma}
\label{lemma:realtwist}
The real structure satisfies
  \begin{equation}
    \mathcal{J}\gamma^\mu=-\gamma^\mu \mathcal J,\quad \mathcal J\omega_\mu^s = \omega_\mu^s \mathcal J,\quad\mathcal J\nabla_\mu^s=+\nabla_\mu^s\mathcal J.
  \end{equation}
\end{lemma}
\begin{proof}
Let us first show that $\J$ anticommutes with the euclidean Dirac
matrices,
\begin{equation}
  \label{eq:113}
  \left\{ {\cal J}, \gamma^\mu_E\right\}=0.
\end{equation}
From the explicit forms \eqref{eq:2bis} of $\cal J$, this is
equivalent to
  \begin{equation}
    \label{eq:100}
    \gamma^0_E\,\gamma^2_E\,\bar\gamma^\mu_E =-\gamma^\mu_E\gamma^0_E\,\gamma^2_E
  \end{equation}
which is true for $\mu=0,2$ since then $\bar\gamma^\mu_E=\gamma^\mu_E$
anticommutes with $\gamma^0_E\gamma^2_E$, and is also true for $\mu=1, 2$
in which case $\bar\gamma^\mu_E =-\gamma^\mu_E$ commutes with
$\gamma^0_E\gamma^2_E$.

Since the spin connection is a real linear combination of products of two
euclidean Dirac matrices, it commutes with $\J$. The latter, having
constant components, commutes with $\partial_\mu$, hence also with the
spin covariant derivative $\nabla_\mu$.

These results hold as well in the curved case, for then one has from \eqref{eq:111}
\begin{equation}
  \label{eq:112}
  \left\{\J, \gamma^\mu\right\}= e_\mu^\alpha \left\{\J, \gamma^\alpha_E\right\}=0.
\end{equation}
\end{proof}

\subsection*{A.2 Components of the gauge sector of the twisted fluctuation}
\label{app:gauge}

The components of the free twisted fluctuation of proposition
\ref{prop:freetwisfluc} are $Z^{r}_\mu = \delta_{\dot s}^{\dot t}
(Z_\mu^r)_{I\alpha}^{J\beta}$ given by (we invert the order of the
lepto-colour and flavour indices in order to make the comparison
with the non-twisted case easier) 
\begin{align}
\left(Z_\mu\right)_{0\dot{1}}^{0\dot{1}}&=2a_\mu,\\
\left(Z_\mu\right)_{0\dot{2}}^{0\dot{2}}&=2 a_\mu+ig_1B_\mu,\\
\left(Z_\mu\right)_{0a}^{0b}&=\delta_a^b\left(w^\mu-  a_\mu\right) +i\left(\delta_a^b \frac{g_1B_\mu}{2}-\frac{g_2}{2}(W_\mu)_a^b\right),\\
\left(Z_\mu\right)_{i\dot{1}}^{j\dot{1}}&= \left(a_\mu
                                            \delta_i^j +(g_\mu)_i^j\right)-i\left(\frac{2g_1
                                            B_\mu}3\delta_i^j+\frac{g_3}{2}(V_\mu)_i^j\right),\\
\left(Z_\mu\right)_{i\dot{2}}^{j\dot{2}}&=\left(a_\mu\delta_i^j
                                            + (g_\mu)_i^j\right)+i\left(\frac{g_1 B_\mu}{3}\delta_i^j-\frac{g_3}{2}(V_\mu)_i^j\right),\\
\left(Z_\mu\right)_{ia}^{jb}&=\left(\delta_a^bw_\mu\delta_i^j-\left(g_\mu\right)_i^j\right)-i\left(\delta_a^b\left(\frac{g_1B_\mu}{6}\delta_i^i+\frac{g_3}{2}(V_\mu)_i^j\right)+\frac{g_2}{2}(W_\mu)_a^b \delta_i^j\right),\\
(Z_\mu)_{I\dot 1}^{J\dot 2}&=(Z^r_\mu)_{I\dot 2}^{J\dot 1}=0.
\end{align}
One then checks that 
\begin{equation*}
 i \begin{pmatrix}
  (Y_\mu)_{i\dot 1}^{j\dot 1} & &\\ &   (Y_\mu)_{i\dot 2}^{j\dot 2}&
  \\ & &   (Y_\mu)_{i\dot a}^{j\dot b}
  \end{pmatrix}_\alpha^\beta=  \begin{pmatrix}-i\left(\frac{2g_1
                                            B_\mu}3\delta_i^j+\frac{g_3}{2}(V_\mu)_i^j\right)& &\\[4pt] & \hspace{-2truecm} i\left(\frac{g_1 B_\mu}{3}\delta_i^j-\frac{g_3}{2}(V_\mu)_i^j\right)&
  \\ & &  \hspace{-2truecm} -i\left(\delta_a^b\left(\frac{g_1B_\mu}{6}\delta_i^i+\frac{g_3}{2}(V_\mu)_i^j\right)+\frac{g_2}{2}(W_\mu)_a^b \delta_i^j\right)
  \end{pmatrix}_\alpha^\beta
\end{equation*}
coincides with the matrix $\mathbb A _\mu^q$ of the non-twisted case
\cite[eq. 1.733]{Connes:2008kx}, while 
 \begin{equation*}
 i \begin{pmatrix}
   (Y_\mu)_{0\dot 1}^{0\dot 1} & &\\ &   (Y_\mu)_{0\dot 2}^{0\dot 2}&
   \\ & &   (Y_\mu)_{0\dot a}^{0\dot b}
   \end{pmatrix}_\alpha^\beta
=  \begin{pmatrix}0& &\\[4pt] 
 & \hspace{0truecm} ig_1 B_\mu& \\ 
 & &  \hspace{-0truecm} i\left(\delta_a^b\frac{g_1B_\mu}{2}-\frac{g_2}{2}(W_\mu)_a^b \right)
  \end{pmatrix}_\alpha^\beta
\end{equation*}
coincides with  the matrix $\mathbb A _\mu^l$ 
\cite[eq. 1.734]{Connes:2008kx}.

\subsection*{A.3 Twisted first-order condition}
\label{sec:twistfirstorder}
 For a twisted spectral
triple, there is a natural twisted version of the  first order
condition \ref{eq:13} that was introduced in \cite{buckley} and whose mathematic
pertinence has been investigated in details in
\cite{Lett.,Landi:2017aa}, namely 
\begin{equation}
  \label{eq:73}
 [[D, b]_\rho, a^\circ]_{\rho^\circ}=0\qquad a,b\in\A
\end{equation}
where $\rho^\circ\in\text{Aut}\,\A^\circ$ is the automorphism of the opposite algebra
$\A^\circ$ induces in \eqref{eq:16} by the twisting automorphism $\rho\in\text{Aut}\,\A$

\begin{proposition}
\label{prop:twisted-first-order}
The free part $\ds\otimes\I_F$ and the diagonal part $\gamma^5\otimes
D_Y$ of the Dirac operator satisfy the twisted first-order condition \eqref{eq:73}, while the off-diagonal part $\gamma^5\otimes D_M$ violates it.
\end{proposition}
\begin{proof}
For  $\ds\otimes\I_F$, using \eqref{eq:55} and
corollary~\ref{cor:twistcommut} one gets
\begin{equation}
\left[\left[\ds\otimes\I_F,b\right]_\rho,JaJ^{-1}\right]_\opr=i\gamma^\mu\begin{pmatrix}
\left[\partial_\mu R,\overline{M}\right]\\&\left[\partial_\mu N,\overline{Q}\right]
\end{pmatrix}_C^D.
\end{equation}
The top-left entry reads (omitting the $s,\dot s$ indices for simplicity)
\begin{equation}
\left[\partial_\mu R,\overline{M}\right]=\begin{pmatrix}
\partial_\mu {\sf d}\left[\I_4,\overline{\sf m}\right]_I^J\\&\partial_\mu p'\left[\I_4,\overline{\sf m'}\right]_I^J
\end{pmatrix}_\alpha^\beta=0.
\end{equation}
Similarly one shows that $\left[\partial_\mu
  N,\overline{Q}\right]=0$, hence $\ds\otimes\I_F$ satisfies the twisted first-order condition. 

For the diagonal part $\gamma^5\otimes D_Y$, lemma \ref{lem:diag1form} and \eqref{eq:55} yield
\begin{equation}
  \label{eq:53}
  \left[\left[\gamma^5\otimes D_Y, b\right]_\rho,
  JaJ^{-1}\right]_{\rho^\circ}=
  -\begin{pmatrix}
    \left[S, \bar
      M\right]_{\rho^\circ}& 0\\
0 & 0
  \end{pmatrix}_C^D.
\end{equation}
In tensorial notations
\begin{align}
  \label{eq:7}
   \left[S, \bar M\right]_{\rho^\circ} &= \delta_{\dot s}^{\dot t}\left[
  \eta_s^u(D_0)_{I\alpha}^{J\gamma}\, R_{u\gamma}^{t\beta},
 \bar M_{s\alpha I}^{t\beta J}  \right]_{\rho^\circ} - \delta_{\dot s}^{\dot t}\left[
\rho(R)_{s\alpha}^{u\gamma}  \eta_u^t(D_0)_{I\gamma}^{J\beta}\,,
 \bar M_{s\alpha I}^{t\beta J}  \right]_{\rho^\circ}.
\end{align}
The right hand side of \eqref{eq:7} is
(omitting the indices $\dot s, \alpha, I$)
\begin{small}
\begin{align}
\nonumber
&\begin{pmatrix}
    D_0 R_r\!\!\!\! \!\!\!\! & \\ &-D_0R_l
  \end{pmatrix}_s^t\!\!\!
\begin{pmatrix}
\bar M_r\!\!\!\! \!\!& \\ 
& \bar M_l
\end{pmatrix}_s^t\!-\! \begin{pmatrix}
\bar M_l\!\!\!\! \!\!& \\ 
& \bar M_r
\end{pmatrix}_s^t \begin{pmatrix}
    D_0R_r\!\!\!\!\!\!\!\! & \\ & -D_0R_l
  \end{pmatrix}_s^t
-\\\nonumber&\qquad-
\begin{pmatrix}
   R_l D_0 \!\!\!\!\!\!\!\!& \\ & -R_r D_0\,
  \end{pmatrix}_s^t
\begin{pmatrix}
\bar M_r\!\!\!\!\!\!& \\ 
& \bar M_l
\end{pmatrix}_s^t +  \begin{pmatrix}
\bar M_l\!\!\!\!& \\ 
& \bar M_r
\end{pmatrix}_s^t \begin{pmatrix}
   R_lD_0\!\!\!\!\!\!\!\!& \\ & -R_rD_0
  \end{pmatrix}_s^t=\\
\label{eq:small}
&=
  \begin{pmatrix}
    D_0R_r\bar M_r- \bar M_lD_0R_r - R_lD_0\bar M_r + \bar M_lR_lD_0\!\!\!\!& \\ & -D_0R_l\bar M_l +
    \bar M_rD_0R_l + R_rD_0 \bar M_l - \bar M_rR_r D_0
  \end{pmatrix}.
\end{align}
\end{small}
From the explicit form \eqref{eq:defD0I} of $D_0$, \eqref{eq:defM} of
$M_{l/r}$ and \eqref{eq:31} of $R_{r/l}$, one checks that
\begin{align}
  \label{eq:36}
  D_0R_r\bar M_r = \bar M_lD_0R_r =
  \begin{pmatrix}
    0 & \bar{\sf k} p'\bar{\sf m'} \\ \sf k \sf d \bar{\sf m}&0
  \end{pmatrix},\;   R_l D_0 \bar M_r = \bar M_l R_l D_0 =
  \begin{pmatrix}
    0 & \bar{\sf k} \bar{\sf m'} \sf d' \\ {\sf k}  \bar{\sf m} p &0
  \end{pmatrix},
\end{align}
so that the upper left term in \eqref{eq:small} is zero. The same is
true for the l.r.t., hence $[S, \bar M]=0$. 

This
shows that \eqref{eq:53} vanishes, which is equivalent to the
proposition.

For the off-diagonal part $\gamma^5\otimes D_M$, one has (omitting the $s,\dot s$ indices for simplicity)
\begin{equation}
\left[\gamma^5\otimes D_M,b\right]_\rho=\begin{pmatrix}
0&\gamma^5\Xi_{I\alpha}^{J\beta}k_R\left(d-d'\right)\\\gamma^5\Xi_{I\alpha}^{J\beta}\overline{k_R}\left(d-d'\right)
\end{pmatrix}_C^D,
\end{equation}
hence
\begin{equation}
\left[\left[\gamma^5\otimes D_M,b\right]_\rho,J\overline{a}J^{-1}\right]_\opr=-\begin{pmatrix}
0&\gamma^5\Xi_{I\alpha}^{K\gamma}k_R\left(d-d'\right){Q}_{K\gamma}^{J\beta}-\rho\left({M}_{I\alpha}^{K\gamma}\right)\gamma^5\Xi_{K\gamma}^{J\beta}k_R\left(d-d'\right)\\\ldots&0
\end{pmatrix}_C^D,
\end{equation}
whose top-right entry reads
\begin{equation}
\gamma^5\Xi_{I\alpha}^{K\gamma}k_R\left(d-d'\right){Q}_{K\gamma}^{J\beta}-\rho\left({M}_{I\alpha}^{K\gamma}\right)\gamma^5\Xi_{K\gamma}^{J\beta}k_R\left(d-d'\right)=k_R\delta_{\dot s}^{\dot t}\Xi_{I\alpha}^{J\beta}\begin{pmatrix}
\sigma+\sigma'\\&-\left(\sigma+\sigma'\right)
\end{pmatrix}_s^t
\end{equation}
and is non-zero.

\end{proof}

  \bibliographystyle{abbrv}
  \bibliography{/Users/pierre/physique/articles/bibdesk/biblio}

\begin{thebibliography}{10}

\bibitem{Besnard:2020ac}
F.~Besnard.
\newblock Extensions of the noncommutative standard model and the weak order
  one condition.
\newblock {\em arXiv 2011.02708}, 2021.

\bibitem{Besnard:2019aa}
F.~Besnard.
\newblock A {U}(1)-{BL} extension of the {S}tandard {M}odel from noncommutative
  geometry.
\newblock {\em J. Math. Phys.}, page 012301, 2021.

\bibitem{Boyle:2019ab}
L.~Boyle and S.~Farnsworth.
\newblock The standard model, the {P}ati-{S}alam model, and {J}ordan geometry.
\newblock (arXiv: 1910.11888), 2019.

\bibitem{T.-Brzezinski:2016aa}
T.~Brzezinski, N.~Ciccoli, L.~Dabrowski, and A.~Sitarz.
\newblock Twisted reality condition for {D}irac operators.
\newblock {\em Math. Phys. Anal. Geo.}, 19(3:16), 2016.

\bibitem{Brzezinski:2018aa}
T.~Brzezinski, L.~Dabrowski, and A.~Sitarz.
\newblock On twisted reality conditions.
\newblock {\em Lett. Math. Phys.}, 109(3):643--659, 04 2018.

\bibitem{Chamseddine:2012fk}
A.~H. Chamseddine and A.~Connes.
\newblock Resilience of the spectral standard model.
\newblock {\em JHEP}, 09:104, 2012.

\bibitem{Chamseddine:2007oz}
A.~H. Chamseddine, A.~Connes, and M.~Marcolli.
\newblock Gravity and the standard model with neutrino mixing.
\newblock {\em Adv. Theor. Math. Phys.}, 11:991--1089, 2007.

\bibitem{Chamseddine:2013uq}
A.~H. Chamseddine, A.~Connes, and W.~van Suijlekom.
\newblock Beyond the spectral standard model: emergence of {P}ati-{S}alam
  unification.
\newblock {\em JHEP}, 11:132, 2013.

\bibitem{Chamseddine:2013fk}
A.~H. Chamseddine, A.~Connes, and W.~van Suijlekom.
\newblock {I}nner fluctuations in noncommutative geometry without first order
  condition.
\newblock {\em J. Geom. Phy.}, 73:222--234, 2013.

\bibitem{Chamseddine:2019aa}
A.~H. Chamseddine and W.~D. van {S}uijlekom.
\newblock A survey of spectral models of gravity coupled to matter.
\newblock {\em In: Chamseddine A., Consani C., Higson N., Khalkhali M.,
  Moscovici H., Yu G. (eds) Advances in Noncommutative Geometry. Springer},
  pages 1--51, 2019.

\bibitem{Connes:1994kx}
A.~Connes.
\newblock {\em Noncommutative Geometry}.
\newblock Academic Press, 1994.

\bibitem{Connes:1996fu}
A.~Connes.
\newblock Gravity coupled with matter and the foundations of noncommutative
  geometry.
\newblock {\em Commun. Math. Phys.}, 182:155--176, 1996.

\bibitem{connesreconstruct}
A.~Connes.
\newblock On the spectral characterization of manifolds.
\newblock {\em J. Noncom. Geom.}, 7(1):1--82, 2013.

\bibitem{Connes:2019ac}
A.~Connes.
\newblock {\em Noncommutative Geometry, the spectral standpoint in "New Spaces
  in Physics: Formal and Conceptual Reflections"}, pages 23--84.
\newblock Cambridge University Press., 10 2021.

\bibitem{Connes:2008kx}
A.~Connes and M.~Marcolli.
\newblock {\em Noncommutative geometry, quantum fields and motives}.
\newblock AMS, 2008.

\bibitem{Connes:1938fk}
A.~Connes and H.~Moscovici.
\newblock Type {III} and spectral triples.
\newblock {\em Traces in number theory, geometry and quantum fields, Aspects
  Math. Friedt. Vieweg, Wiesbaden}, E38:57--71, 2008.

\bibitem{Devastato:2018aa}
A.~Devastato, S.~Farnsworth, F.~Lizzi, and P.~Martinetti.
\newblock Lorentz signature and twisted spectral triples.
\newblock {\em JHEP}, 03(089), 2018.

\bibitem{Devastato:2013fk}
A.~Devastato, F.~Lizzi, and P.~Martinetti.
\newblock {G}rand {S}ymmetry, {S}pectral {A}ction and the {H}iggs mass.
\newblock {\em JHEP}, 01:042, 2014.

\bibitem{buckley}
A.~Devastato and P.~Martinetti.
\newblock Twisted spectral triple for the standard model and spontaneous
  breaking of the grand symmetry.
\newblock {\em Math. Phys. Anal. Geo.}, 20(2):43, 2017.

\bibitem{Filaci:2021aa}
M.~Filaci and P.~Martinetti.
\newblock A minimal twist for the {S}tandard {M}odel in noncommutative geometry
  ii: the fermionic action.
\newblock {\em In preparation}, 2021.

\bibitem{Manuel-Filaci:2020aa}
M.~Filaci and P.~Martinetti.
\newblock Minimal twist of almost commutative geometries.
\newblock {\em In preparation}, 2021.

\bibitem{Lett.}
G.~Landi and P.~Martinetti.
\newblock On twisting real spectral triples by algebra automorphisms.
\newblock {\em Lett. Math. Phys.}, 106:1499--1530, 2016.

\bibitem{Landi:2017aa}
G.~Landi and P.~Martinetti.
\newblock Gauge transformations for twisted spectral triples.
\newblock {\em Lett. Math. Phys.}, 12:2589--2626, 2018.

\bibitem{Martinetti:2019aa}
P.~Martinetti and D.~Singh.
\newblock Lorentzian fermionic action by twisting euclidean spectral triples.
\newblock Preprint: arXiv 1907.02485, 2019.

\bibitem{Martinetti:2021aa}
P.~Martinetti and J.~Zanchettin.
\newblock Twisted spectral triples without the first order condition.
\newblock {\em arXiv 2103.15643}, 2021.

\bibitem{Walterlivre}
W.~van Suijlekom.
\newblock {\em Noncommutative geometry and particle physics}.
\newblock Springer, 2015.

\end{thebibliography}

\end{document}